\documentclass[10pt,twocolumn,journal]{IEEEtran}

\usepackage{listings}
\usepackage{amsmath}
\usepackage{amsthm}
\usepackage{tikz}
\usepackage{caption}
\usepackage{array}
\usepackage{mdwmath}
\usepackage{multirow}
\usepackage{mdwtab}
\usepackage{eqparbox}
\usepackage{multicol}
\usepackage{amsfonts}
\usepackage{tikz}
\usepackage{multirow,bigstrut,threeparttable}
\usepackage{amsthm}
\usepackage{array}
\usepackage{bbm}
\usepackage{subfigure}
\usepackage{epstopdf}
\usepackage{mdwmath}
\usepackage{mdwtab}
\usepackage{eqparbox}
\usetikzlibrary{topaths,calc}
\usepackage{latexsym}
\usepackage{cite}
\usepackage{amssymb}
\usepackage{bm}
\usepackage{amssymb}
\usepackage{graphicx}
\usepackage{mathrsfs}
\usepackage{epsfig}
\usepackage{psfrag}
\usepackage{setspace}
\usepackage{balance}
\usepackage{url}
\usepackage{algorithm}
\usepackage{algpseudocode}
\usepackage{stfloats}

\usepackage{mathtools}
\usepackage{blkarray}
\usepackage{multirow,bigdelim,dcolumn,booktabs}

\usepackage{xparse}
\usepackage{tikz}
\usetikzlibrary{calc}
\usetikzlibrary{decorations.pathreplacing,matrix,positioning}

\usepackage[T1]{fontenc}
\usepackage[utf8]{inputenc}
\usepackage{mathtools}
\usepackage{blkarray, bigstrut}
\usepackage{gauss}

\newcommand*{\BraceAmplitude}{0.4em}%
\newcommand*{\VerticalOffset}{0.5ex}%
\newcommand*{\HorizontalOffset}{0.0em}%
\newcommand*{\blocktextwid}{3.0cm}%
\NewDocumentCommand{\InsertLeftBrace}{%
	O{} 
	O{\HorizontalOffset,\VerticalOffset} 
	O{\blocktextwid} 
	m   
	m   
	m   
}{%
	\begin{tikzpicture}[overlay,remember picture]
	\coordinate (Brace Top)    at ($(#4.north) + (#2)$);
	\coordinate (Brace Bottom) at ($(#5.south) + (#2)$);
	\draw [decoration={brace, amplitude=\BraceAmplitude}, decorate, thick, draw=black, #1]
	(Brace Bottom) -- (Brace Top) 
	node [pos=0.5, anchor=east, align=left, text width=#3, color=black, xshift=\BraceAmplitude] {#6};
	\end{tikzpicture}%
}%
\NewDocumentCommand{\InsertRightBrace}{%
	O{} 
	O{\HorizontalOffset,\VerticalOffset} 
	O{\blocktextwid} 
	m   
	m   
	m   
}{%
	\begin{tikzpicture}[overlay,remember picture]
	\coordinate (Brace Top)    at ($(#4.north) + (#2)$);
	\coordinate (Brace Bottom) at ($(#5.south) + (#2)$);
	\draw [decoration={brace, amplitude=\BraceAmplitude}, decorate, thick, draw=black, #1]
	(Brace Top) -- (Brace Bottom) 
	node [pos=0.5, anchor=west, align=left, text width=#3, color=black, xshift=\BraceAmplitude] {#6};
	\end{tikzpicture}%
}%
\NewDocumentCommand{\InsertTopBrace}{%
	O{} 
	O{\HorizontalOffset,\VerticalOffset} 
	O{\blocktextwid} 
	m   
	m   
	m   
}{%
	\begin{tikzpicture}[overlay,remember picture]
	\coordinate (Brace Top)    at ($(#4.west) + (#2)$);
	\coordinate (Brace Bottom) at ($(#5.east) + (#2)$);
	\draw [decoration={brace, amplitude=\BraceAmplitude}, decorate, thick, draw=black, #1]
	(Brace Top) -- (Brace Bottom) 
	node [pos=0.5, anchor=south, align=left, text width=#3, color=black, xshift=\BraceAmplitude] {#6};
	\end{tikzpicture}%
}%

\usetikzlibrary{patterns}

\definecolor{cof}{RGB}{219,144,71}
\definecolor{pur}{RGB}{186,146,162}
\definecolor{greeo}{RGB}{91,173,69}
\definecolor{greet}{RGB}{52,111,72}

\theoremstyle{plain}
\newtheorem{theorem}{Theorem}

\newtheorem{lemma}{Lemma}
\newtheorem{remark}{Remark}
\newtheorem{corollary}{Corollary}
\newtheorem{definition}{Definition}

\newtheorem*{folklore}{Folklore}

\def \bP {\mathbb{P}}

\def \bR {\mathbb{R}}

\def\1{\mathbbm{1}}
\def \FF {\mathbb{F}}

\usepackage{xspace}

\newcommand{\stepa}[1]{\overset{\rm (a)}{#1}}
\newcommand{\stepb}[1]{\overset{\rm (b)}{#1}}
\newcommand{\stepc}[1]{\overset{\rm (c)}{#1}}
\newcommand{\stepd}[1]{\overset{\rm (d)}{#1}}
\newcommand{\stepe}[1]{\overset{\rm (e)}{#1}}

\newcommand{\reals}{\mathbb{R}}

\definecolor{myblue}{rgb}{.8, .8, 1}
\definecolor{mathblue}{rgb}{0.2472, 0.24, 0.6} 
\definecolor{mathred}{rgb}{0.6, 0.24, 0.442893}
\definecolor{mathyellow}{rgb}{0.6, 0.547014, 0.24}

\newcommand{\blue}{\color{black}}
\newcommand{\nb}[1]{{\blue #1}}

\newcommand{\calC}{{\mathcal{C}}}

\newcommand{\calS}{{\mathcal{S}}}

\usepackage{cleveref}
\crefname{lemma}{Lemma}{Lemmas}
\Crefname{lemma}{Lemma}{Lemmas}
\crefname{thm}{Theorem}{Theorems}
\Crefname{thm}{Theorem}{Theorems}
\crefformat{equation}{(#2#1#3)}

\begin{document}

\title{Optimal Communication Rates and Combinatorial Properties for Common Randomness Generation}
\author{Yanjun Han, \IEEEmembership{Member,~IEEE}, Kedar Tatwawadi, Gowtham R. Kurri, \IEEEmembership{Member,~IEEE}, Zhengqing Zhou, Vinod M. Prabhakaran, \IEEEmembership{Member,~IEEE}, and Tsachy Weissman, \IEEEmembership{Fellow,~IEEE}\thanks{Yanjun Han and Kedar Tatwawadi contribute equally to this paper. The work of Gowtham R. Kurri and Vinod M. Prabhakaran was supported by the Department of Atomic Energy, Government of India, under Project RTI4001. The work of Vinod M. Prabhakaran was also supported by the Science \& Engineering Research Board, India through project MTR/2020/000308. This article was presented in part at the 2021 IEEE International Symposium on Information Theory (ISIT).

Yanjun Han was with the Department of Electrical Engineering, Stanford University, Stanford, CA 94305 USA. He is now with the Simons Institute for the Theory of Computing, University of California at Berkeley, Berkeley, CA 94720 USA (e-mail: yjhan@berkeley.edu).

Kedar Tatwawadi was with the Department of Electrical Engineering, Stanford University, Stanford, CA 94305 USA. He is now with WaveOne Inc. as a Research Scientist. (email: kedart@wave.one).

Gowtham R. Kurri was with the Tata Institute of Fundamental Research, India. He is now with the School of Electrical, Computer and Energy Engineering, Arizona State University, Tempe, AZ 85287 USA (email: gowthamkurri@gmail.com).

Zhengqing Zhou is with the Department of Mathematics, Stanford University, Stanford, CA 94305 USA (email: zqzhou@stanford.edu). 

Vinod Prabhakaran is with the School of Technology and Computer Science, Tata Institute of Fundamental Research, Mumbai 400005, India (email: vinodmp@tifr.res.in). 

Tsachy Weissman is with the Department of Electrical Engineering, Stanford University, Stanford, CA 94305 USA (email: tsachy@stanford.edu).  

Copyright \copyright \ 2021 IEEE. Personal use of this material is permitted.  However, permission to use this material for any other purposes must be obtained from the IEEE by sending a request to pubs-permissions@ieee.org.
}}

\maketitle

\begin{abstract}
We study \nb{common randomness generation} problems where $n$ players aim to generate \emph{same} sequences of random coin flips where some subsets of the players share an independent common coin which can be tossed multiple times, and there is a publicly seen blackboard through which the players communicate with each other. We provide a tight representation of the optimal communication rates via linear programming, and more importantly, propose explicit algorithms for the optimal distributed simulation for a wide class of hypergraphs. In particular, the optimal communication rate in complete hypergraphs is still achievable in sparser hypergraphs containing a path-connected cycle-free cluster of topologically connected components. Some key steps in analyzing the upper bounds rely on two different definitions of connectivity in hypergraphs, which may be of independent interest. 
\end{abstract}

%
%


\section{Introduction}
Common randomness, or shared randomness, refers to some external randomness known to all agents which enables them to take coordinated actions. The most classical application of common randomness is the generation of the secret key in cryptography \cite{AhlswedeC93}. This is also a valuable resource which aids diverse applications including developing randomized algorithms \cite{mitzenmacher2005probability}, reducing the communication complexity in distributed computing \cite{kushilevitz1996communication}, reducing the sample complexity in distributed inference \cite{acharya2019domain}, coordination among players in game theory \cite{AnantharamB07}, and quantum mechanics \cite{bennett2002entanglement}. In these applications, generating common randomness, or distributed simulation of the same random sequence, is of the utmost importance. 

In many scenarios, there is shared randomness within certain subsets of the agents, and sound communication strategies are necessary to generate common randomness for all agents. Consider the following simple example: Alice shares independent randomness with Bob and Carlo respectively, and Alice aims to broadcast as few messages as possible to Bob and Carlo so that they have access to some common randomness. The simplest strategy for Alice is to broadcast any random bit $R_0$, then they generate $1$ bit of common randomness with $1$ bit of communication. However, if Alice broadcasts $R_1 \oplus R_2$ where the bits $R_1$ and $R_2$ come from the shared randomness with Bob and Carlo, respectively, then they successfully generate $2$ bits of common randomness still with $1$ bit of communication (see Appendix \ref{subsec.k=2_easy} for more details). Hence, the communication resources may be saved under better strategies. 

In this paper, we consider a natural generalization of the above scenario: we are given a hypergraph $G=(V,E)$, where the vertex set $V=[n]$ is the set of $n$ players, and the edge set $E=\{e_1,\cdots,e_m\}$ consists of hyperedges $e_i\subseteq V$ representing the subsets of players sharing a common fair coin. We assume that the coins for different hyperedges are mutually independent. The players can toss the shared coins multiple times as a part of the communication strategy. \nb{In particular, the number of coin tosses for each hyperedge is not pre-determined and this allows for the scenario where different hyperedges could be used different times depending on the structure of the hypergraph.} We also assume that the players may communicate with each other via a blackboard communication protocol \cite{kushilevitz1997communication}, i.e. each player may write some messages on a publicly seen blackboard based on his shared coins and all current message on the blackboard. The blackboard communication protocol allows for interactive strategies and is stronger than both the \emph{simultaneous message passing} (SMP) protocol where each player writes messages on the blackboard independently of each other, and the sequential message passing protocol where players write messages sequentially but in a fixed order. The objective of the players is to generate the \emph{same} random variable (or vector) $X$ following a given target discrete distribution while minimizing the communication cost, i.e. the entropy of the message $M$ written on the blackboard. We define the communication rate as the ratio $H(M) / H(X)$, where $H(\cdot)$ denotes the Shannon entropy of discrete random variables. We provide a tight representation of the optimal communication rates via linear programming (see Theorem~\ref{thm.main.lower} and discussions followed). More importantly, we also propose explicit algorithms \nb{and investigate combinatorial properties} for the optimal \nb{common randomness generation} for a wide class of hypergraphs (Theorem~\ref{thm.main.upper}).

\subsection{Related works}

\nb{The role of common randomness (CR) has been given considerable attention in information theory literature starting from G\'{a}cs and K\"{o}rner \cite{gacs1973common} who characterized the maximum rate of common randomness that can be extracted from a pair of correlated random variables. Wyner~\cite{Wyner75} characterized the minimum rate of CR required for two processors to produce (approximately) independent copies of correlated random variables. CR was used for encoding and decoding in arbitrary varying channels by Ahlswede~\cite{Ahlswede78}, and Csisz\'ar and Narayan \cite{CsiszarN88}. CR generation with interactive communication between two players was studied by Ahlswede and Csisz\'{a}r~\cite{AhlswedeC98}. CR generation with a helper was studied by Csisz\'{a}r and Narayan~\cite{CsiszarN00}. CR generation via a network of discrete memoryless channels was studied by Venkatesan and Anantharam~\cite{VenkatesanA00}. Zhao and Chia~\cite{ZhaoC11} studied the relation between Hirschfeld-Gebelein-Renyi maximal correlation and CR generation. CR generation between two players which should be hidden from an eavesdropper was studied in secret key (SK) agreement by Maurer~\cite{Maurer93}, and Ahlswede and Csisz\'{a}r~\cite{AhlswedeC93}. Secret key agreement between multiple players was studied by Csisz\'ar and Narayan~\cite{CsiszarN04}. This is closely related to communication for omnicience~\cite{nitinawarat2010perfect,DingCQRS18}. The minimum communication rate required to generate secret key between two players was studied by Tyagi~\cite{Tyagi13}, and Ghazi and Jayram~\cite{ghazi2018resource}. Liu et al.~\cite{LiuCV17} characterized the trade-off between secret key and communication rates for a fixed number of communication rounds. Building on Tyagi~\cite{Tyagi13}, Mukherjee et al.~\cite{mukherjee2016public} derived a lower bound on this communication rate for SK agreement in the multiterminal source model.}

 \nb{A special source model, i.e. the \emph{hypergraphical source model}~\cite{chan2010mutual,Rouayhebetal}, where clusters of players share independent randomness, has received attention in various works which studied SK capacity as a function of the total communication rate~\cite{mukherjee2016public,CourtadeH16,ChanMKZ18,ZhouC20,chan2019secret}.  Courtade and Halford~\cite{CourtadeH16} considered the non-asymptotic one-shot version of the SK generation problem and characterized the minimum amount of communication needed under an assumption that communication is a linear function of the sources. Chan et al.~\cite{ChanMKZ18} studied the optimality of SK agreement via omniscience. Zhou and Chan~\cite{ZhouC20} studied minimally connected hypergraphs and characterized the optimal trade-off between secret key rate and communication rate tuple. Chan~\cite{chan2019secret} characterized a similar achievable rate region for any general hypergraph in terms of a polynomial-time computable linear program. Hypergraphical source model is a generalization of the Pairwise Independent Network (PIN) Model, where every pair of players share independent randomness, first introduced by Ye and Reznik~\cite{YeR07} and studied in \cite{Nitinawaratetal10,nitinawarat2010perfect,Chanetal19,mukherjee2016public}. Our work is also on the hypergraphical source model, but differs from the previous works in that we exploit the combinatorial nature of general hypergraphs. We remark that the hypergraph theory plays an important role in Theorem \ref{thm.main.upper}. Specifically, the two different notions of hypergraph connectivity presented in Theorem~\ref{thm.main.upper} aim to generalize the following folklore in different ways (see Lemmata~\ref{thm.connectivity} and \ref{lemma.path_connectivity_edge}):} 
\begin{folklore}
	A tree on $n$ vertices has exactly $n-1$ edges.
\end{folklore} 
For $k\ge 3$, a proper definition of trees in hypergraphs is required to generalize the above folklore. Recall that a tree enjoys two essential properties, i.e., \emph{connectivity} and \emph{cycle-free}, therefore a proper definition of connectivity is important. In combinatorics, the most common definition of connectivity is the path connectivity or its variants\cite{matrixTree3uniform, minimunSpanning2003, spanningTree3uniform}, which imposes constraints on \emph{vertices} and requires that any two vertices can reach each other through the $1$-dimensional skeleton of the hyperedges. Consequently, the cycle-free property can also be defined in terms of paths (cycles). There is also another less famous notion of hypergraph connectivity due to Kalai \cite{Kalai1983} which imposes constraints on the \emph{facets} of the hypergraph and requires them to be connected topologically. In the language of algebraic topology, a $k$-uniform hypergraph can be treated as a $(k-1)$-dimensional simplicial complex $\calC$, with the facets being the hyperedges. Then the hypergraph is topologically connected if and only if the $(k-2)$-skeleton of $\calC$ is full. The cycle-free property can then be defined as that the $(k-1)$-th simplicial homology of $\calC$ is 0 \cite{Kalai1983, simplicialMatrixTree}. From both directions we may obtain appropriate generalizations of the previous folklore (see Lemmas \ref{thm.connectivity} and \ref{lemma.path_connectivity_edge}, respectively), which constitute the key ingredients of Theorem \ref{thm.main.upper}.

\nb{The work by Mukherjee et al.~\cite{mukherjee2016public} deserves special mention. Specifically, it showed that if the $k$-uniform hypergraph, or in general any multiterminal source model, is \emph{of type $\calS$} (a notion introduced in \cite{mukherjee2016public}), then there is a strategy achieving the optimal communication rate $\frac{n-k}{n-1}$ and outputting each hyperedge (from a multi-hypergraph) exactly once. The main differences between our work and \cite{mukherjee2016public} are as follows. First, our achievability scheme is non-asymptotic (i.e. no blocklengths required) and combinatorial, while the scheme in \cite{mukherjee2016public} potentially requires large blocklengths and is more information-theoretic. Second, although the type $\calS$ condition is a nice ``if and only if'' result and could be checked efficiently in polynomial time for a given hypergraph (see also \cite{chan2015multivariate}), a rich combinatorial characterization about which family of hypergraphs are of type $\calS$ remains unclear. Our work aims to provide a partial answer to this combinatorial problem, and based on the fundamental notions of connectivity, proposes rich families of hypergraphs that achieve the optimal $\frac{n-k}{n-1}$ communication rate. Although our families of hypergraphs must be of type $\calS$, it is worth noting that so far we do not have a direct argument to connect them. Thus, our work presents an alternative approach which sheds more lights on the combinatorial perspective. }

We also review some literature on the communication complexity. First introduced in \cite{yao1979some}, the blackboard communication protocol serves as an elegant mathematical framework for the study of communication complexity. A series of research is devoted to the lower bounds in communication complexity, where the log rank is the prominent tool for all the deterministic \cite{mehlhorn1982vegas,yannakakis1991expressing}, nondeterministic \cite{karchmer1992fractional} and randomized communication complexities \cite{yao1983lower,newman1991private,krause1996geometric}. We refer to \cite{kushilevitz1996communication} for a survey of these methods. Another closely-related problem is distributed inference under communication constraints \cite{zhang2013information}, where distributed simulation of common randomness is useful for distributed learning and property testing \cite{acharya2018distributed,acharya2018inference}. To establish lower bounds on the communication complexity in distributed inference, the copy-paste property of the blackboard communication model typically plays an important role \cite{braverman2016communication,han2018geometric}. However, our technique to establish the lower bound is different, where only the sequential nature of the blackboard communication protocol is used in the proof of Theorem \ref{thm.main.lower}, which may be of independent interest. 

\section{Main Results}\label{Mainresults}
The first theorem presents a general lower bound of the communication rate for any hypergraph. 
\begin{theorem}\label{thm.main.lower}
	Let $G=(V,E)$ be any hypergraph. Let $X$ be the discrete random variable outputted by each vertex through a blackboard communication protocol, and $M$ be the message written on the blackboard. Then $H(M) / H(X) \ge t(G)$, where $t(G)$ is the solution to the following linear program: 
	\begin{align*}
	t(G) = \begin{cases}
	\min \qquad \sum_{v\in V} r_v, \\
	\text{\rm subject to } \qquad \sum_{v\in U} r_v \ge \sum_{e\in E: e\subseteq U} s_e, \quad \forall U\subsetneq V, \\
	\qquad \qquad \qquad \quad \sum_{e\in E} s_e \ge 1, \\
	\qquad \qquad \qquad \quad r_v, s_e \ge 0, \quad \forall v\in V, e\in E. 
	\end{cases}
	\end{align*}
\end{theorem}
\nb{A detailed proof of Theorem~\ref{thm.main.lower} is in Appendix~\ref{theorem1}. The linear program in Theorem~\ref{thm.main.lower} can be seen as a special case of a linear program~\cite[Corollary~2]{chan2019secret} (see also \cite{DingCQRS18}) in a closely related problem of secret-key agreement where it is also shown to be solvable in polynomial time. In fact, \cite[Corollary~2]{chan2019secret} implies the result in Theorem~\ref{thm.main.lower}}\footnote{\nb{We thank Chung Chan for pointing out to us that Theorem~\ref{thm.main.lower} follows from \cite[Corollary~2]{chan2019secret} and the fact that the associated linear program is solvable in polynomial-time. We note that Theorem~\ref{thm.main.lower} appeared in a version of the current paper~\cite{HanTZKPW20} (arXiv:1904.03271v2) slightly earlier than \cite{chan2019secret}~(arXiv:1910.01894v1) but without the observation of polynomial-time solvability.}}. 
Intuitively, the quantity $r_v$ denotes the length of the messages sent by player $v$, and $s_e$ denotes the number of random bits extracted from the hyperedge $e$ to generate the common output $X$. Therefore, the first inequality constraints require that for any graph cut $U\subsetneq V$, the amount of information communicated from the players in $U$ should at least cover the amount of randomness extracted out of hyperedges totally contained in $U$. These constraints also turn out to be tight in the sense that the optimal communication rate $t(G)$ can be attained asymptotically (as $H(X)$ goes to infinity) via {\blue linear network coding~\cite{nitinawarat2010perfect} - see Appendix~\ref{asymptachv} for details.}

Although Theorem \ref{thm.main.lower} (together with the asymptotic upper bounds) provides a tight characterization of the optimal communication rates for common randomness generation, the picture is still incomplete due to the following reasons. 
First, the existential proof of the network coding approach in Appendix~\ref{asymptachv} does not give an explicit communication strategy, and the result is asymptotic in the sense that large blocklengths are required and the communication rate only approaches but may never reach $t(G)$. Second, the linear program tells little about the combinatorial properties of the hypergraphs where a small communication rate is possible. For example, which hypergraphs are as good as the complete graphs? 

To answer these questions, in this paper we propose explicit algorithms of communication strategies and investigate the combinatorial properties of hypergraphs which lead to a small communication rate, at the expense of losing certain generalities. \nb{First we investigate some basic properties of $t(G)$ for general hypergraphs.
\begin{corollary}\label{cor.upper}
It always holds that $t(G)\le 1$ for any hypergraph $G$, with equality if and only if $G$ is disconnected (in the usual sense of path connectivity formally defined in Definition \ref{def.path_connectivity}).
\end{corollary}

A proof of Corollary \ref{cor.upper} is given in Appendix \ref{corollary1proof}. Next we turn to the lower bound of $t(G)$, and investigate} the hypergraph structures which perform equally well as the complete $k$-uniform hypergraphs. Note that a hypergraph $G=(V,E)$ is called $k$-uniform if for all hyperedges $e\in E$ we have $|e|=k$. The following corollary follows immediately from Theorem \ref{thm.main.lower}. 

\begin{corollary}\label{cor.lower}
If $G=(V,E)$ is a $k$-uniform hypergraph, then
	$$
	t(G) \ge \frac{n-k}{n-1}. 
	$$
\end{corollary}

\nb{A proof of Corollary \ref{cor.lower} is given in Appendix~\ref{corollary2proof}.} By Corollary \ref{cor.lower}, it remains to find hypergraph structures and explicit communication strategies where the optimal rate $(n-k)/(n-1)$ is achievable. \nb{It turns out that the simple graph case $k=2$ admits an explicit characterization of $t(G)$. 

\begin{corollary}\label{cor.simplegraph}
If $G$ is a simple graph (i.e. $2$-uniform), then 
\begin{align*}
	t(G) = \begin{cases}
		1 &\text{if } G\text{ is not connected}, \\
		\frac{n-2}{n-1} &\text{if } G\text{ is connected}. 
	\end{cases}
\end{align*}
\end{corollary}
In Corollary \ref{cor.simplegraph}, the case of disconnected graphs follows from Corollary \ref{cor.upper}, and that of connected graphs follows from the lower bound of $t(G)$ in Corollary \ref{cor.lower} and an explicit achievability strategy in Appendix \ref{sec.example}. Therefore, both Corollaries \ref{cor.upper} and \ref{cor.simplegraph} show that hypergraph connectivity plays a central role in achieving a small communication rate $t(G)$, and one may wonder whether the lower bound of Corollary \ref{cor.lower} is achievable whenever the hypergraph is connected. 
}However, this does not generalize to any $k$-uniform hypergraphs with $k\ge 3$ under the usual notion of path connectivity for graphs, and a number of path-connected hypergraphs are too sparse to achieve a small communication rate. It also becomes challenging to propose an achievability scheme even if $k=3$. The following theorem shows that under the correct definitions of connectivity, the optimal rate of communication is attainable. 

\begin{theorem}\label{thm.main.upper}
Let $G=(V,E)$ be a $k$-uniform hypergraph, with $1\le k\le n$. If $G$ is a path-connected cycle-free cluster (cf. Definition \ref{def.cluster}) of topologically connected components (cf. Definition \ref{def.connectivity}), then there exists an explicit communication strategy under the simultaneous message passing protocol such that for some $m\in \mathbb{N}$, each vertex can output the same random vector $X\sim \mathsf{Unif}(\{0,1\}^m)$ while the message $M$ written on the blackboard satisfies
\begin{align*}
\frac{H(M)}{H(X)} = \frac{n-k}{n-1}. 
\end{align*}
\end{theorem}

\begin{remark}
Although Theorem \ref{thm.main.upper} restricts the output $X$ to be an independent and identically distributed (i.i.d.) Bernoulli random vector, the same communication rate can also be generalized to any i.i.d. random vectors in an asymptotic manner.
{\blue This is precisely because a common randomness of rate $H(X)$ suffices to generate i.i.d. copies of a random variable $X$ with asymptotically (in the number of shared coin tosses) vanishing Kullback-Leibler divergence or total variation distance \cite{Wyner75,knuth1976complexity,Cuff13}.}  	
\end{remark}

{\blue A detailed description and proof of Theorem~\ref{thm.main.upper} are deferred to Sections~\ref{sec.achievability} and \ref{sec.generalization}.} Theorem \ref{thm.main.upper} shows that the optimal rate $(n-k)/(n-1)$ is attainable non-asymptotically when the underlying hypergraph satisfies suitable connectivity conditions, which are generalizations of the classical connectivity for $k=2$ from two different angles. We remark that a path-connected cycle-free cluster of topologically connected components differs significantly from the usual notion of path connectivity in hypergraphs, where the topological connectivity, the central concept in Theorem \ref{thm.main.upper} and a stronger notion than path connectivity, views the hypergraph as a simplicial complex in the context of algebraic topology. For example, when $k=3$ and $n=4$, the hyperedges may be viewed as surfaces of a pyramid; two surfaces suffice to make the hypergraph path-connected, while three surfaces are necessary to make it topologically connected. We leave more discussions to the related works on hypergraph theory and formal definitions in Section \ref{sec.achievability}.

The new notion of connectivity contains a rich family of hypergraphs which suggests that Theorem \ref{thm.main.upper} covers all hypergraphs for which the optimal communication rate $(n-k)/(n-1)$ is achievable. Surprisingly, there are indeed richer families of hypergraphs which do not follow the previous connectivity notion but still achieve the optimal communication rate. We discuss these examples in Section \ref{subsec.non-examples}, where we characterize the complete class of optimal hypergraphs in certain cases such as $k=2$, and $k=3$ \emph{star-shaped} hypergraphs, which are discussed in Appendix \ref{appendix.star_graphs}. It is an outstanding open problem to figure out the complete class of optimal hypergraphs. 

\subsection{Organization}
The rest of this paper is organized as follows. Section \ref{sec.achievability} gives the formal definition of topological connectivity in $k$-uniform hypergraphs and proposes the optimal communication strategy on topologically $k$-connected hypergraphs, and Section \ref{sec.generalization} generalizes the path connectivity and presents a general algorithm for Theorem \ref{thm.main.upper}. Proofs of main results are deferred to the appendices, where Appendix \ref{sec.example} also provides examples where the achievability scheme is comparatively simple, including the complete picture of \nb{$k$-uniform hypergraphs with} $k=2$. 

\subsection{Notations}
Let $\mathbb{N}$ be the set of all non-negative integers, and $\FF_2$ be the binary field. We denote by $\oplus$ the addition operator in $\FF_2$, and for $n\in \mathbb{N}$, we denote $[n]\triangleq \{1,2,\cdots,n\}$. For discrete random variables $X,Y$, let $H(X)$ be the Shannon entropy of $X$ (in bits), and $I(X;Y)$ be the mutual information between $X$ and $Y$. For a set $A$ and $k\in \mathbb{N}$, let $|A|$ be the cardinality of $A$, and $\binom{A}{k}$ be the collection of all size-$k$ subsets of $A$. Consequently, a $k$-uniform hypergraph $G=(V,E)$ is complete if $E=\binom{V}{k}$.

\section{Achievability: Topological Connectivity}\label{sec.achievability}
In this section we provide an achievability scheme for general topologically $k$-connected hypergraphs. We introduce the definition and properties of topological connectivity in Section \ref{subsec.connectivity} and the corresponding achievability strategy in Section \ref{subsec.strategy}. 

\subsection{Topological connectivity}\label{subsec.connectivity}
In Appendix \ref{subsec.k=2_general}, general achievability schemes have been proposed for all connected simple graphs when $k=2$. A natural conjecture would be that similar ideas should also work for general ``connected'' $k$-uniform hypergraphs. We will show that this conjecture is true, while we need the correct definition of connectivity for $k$-uniform hypergraphs.

In our paper, we adopt the tree definition in \cite{Kalai1983} and reinterpret it as \emph{topological connectivity}: 
\begin{definition}[Topologically $k$-connected hypergraph]\label{def.connectivity}
For any $k$-uniform hypergraph $G=(V,E)$ with $k\ge 2$, define the following \emph{generation step}: for hyperedges $e_1,\cdots,e_m\in E$ and any hyperedge $e\notin E$, if all $(k-1)$-tuples in $\binom{V}{k-1}$ appearing in $e_1,\cdots,e_m,e$ appear an even number of times, we may add the hyperedge $e$ to the hypergraph. We call $G$ is \emph{topologically $k$-connected} if $G$ becomes a complete $k$-uniform hypergraph after a finite number of generation steps. 
\end{definition}

\begin{definition}[Minimal topologically $k$-connected hypergraph]\label{def.minimal}
For $k\ge 2$, a $k$-uniform hypergraph $G$ is called \emph{minimal topologically $k$-connected} if $G$ is topologically $k$-connected and removing any hyperedge of $G$ makes it become not topologically $k$-connected. 
\end{definition}

The generation step has a natural topological interpretation. Think of embedding the $k$-uniform hypergraph $G$ into $\reals^k$, and treat hyperedges of $G$ as $(k-1)$-dimensional \emph{facets} (cf. Figure \ref{fig.connected_3hypergraph}). Note that the technical condition that all $(k-1)$-tuples appearing in $e_1,\cdots,e_m,e$ appear an even number of times essentially says that the faces $e_1,\cdots,e_m,e$ form the closed surface of a polygon. Then the generation step states that, if there is a $k$-dimensional polygon with all but one faces in the hypergraph, we are allowed to add this missing face to the hypergraph. When $k=2$, this definition coincides with the usual path-connectivity for undirected graphs, where we are allowed to add an edge $(u,v)$ to form a cycle (i.e. a $2$-dimensional polygon) if there is a path from $u$ to $v$. 

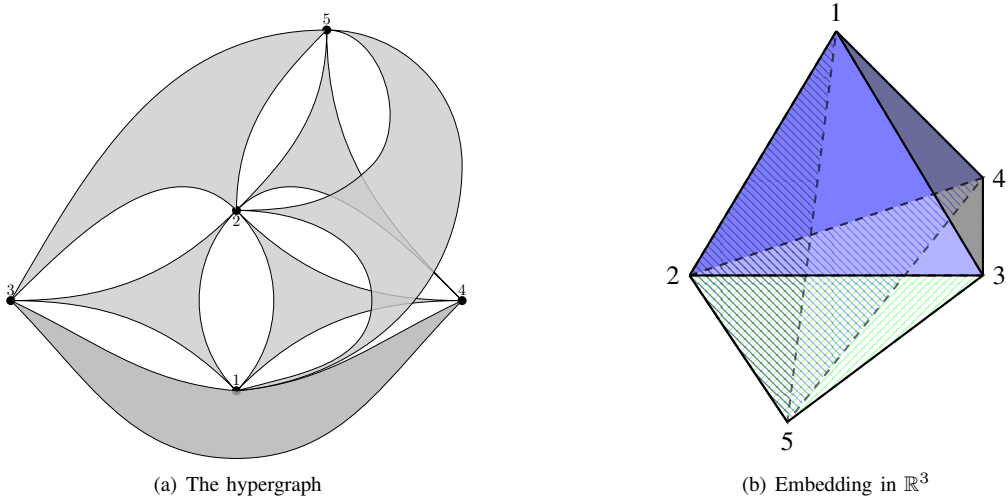
\begin{figure*}[htbp]
\centering
\subfigure[The hypergraph]{\scalebox{0.6}{
	\begin{tikzpicture}
	\node (1) at (6,8) {};
	\node (2) at (6,12) {};
	\node (3) at (1,10) {};
	\node (4) at (11,10) {};
	\node (5) at (8,16) {};
	
	\fill (1) circle (0.1) node [above] {$1$};
	\fill (2) circle (0.1) node [below] {$2$};
	\fill (3) circle (0.1) node [above] {$3$};
	\fill (4) circle (0.1) node [above] {$4$};
	\fill (5) circle (0.1) node [above] {$5$};
	
	\begin{scope}[fill opacity=0.8]
	\filldraw[fill=black!20] ($(1)$) 
	to[out=135,in=225]($(2)$)
	to[out=225,in=0]($(3)$)
	to[out=0,in=135]($(1)$)
	;
	\filldraw[fill=black!20] ($(1)$) 
	to[out=45,in=315]($(2)$)
	to[out=315,in=180]($(4)$)
	to[out=180,in=45]($(1)$)
	;
	\filldraw[fill=black!30] ($(3)$) 
	to[out=-25,in=175]($(1)$)
	to[out=5,in=205]($(4)$)
	to[out=225,in=0]($(1) + (0,-1.5)$)
	to[out=180,in=-45]($(3)$)
	;
	\filldraw[fill=black!20] ($(5)$) 
	to[out=270,in=45]($(2)$)
	to[out=45,in=135]($(4)$)
	to[out=135,in=270]($(5)$)
	;
	\filldraw[fill=black!20] ($(5)$) 
	to[out=225,in=90]($(2)$)
	to[out=135,in=45]($(3)$)
	to[out=135,in=45]($(3)$)
	to[out=60,in=180]($(5)$)
	;
	\filldraw[fill=black!20] ($(1)$) 
	to[out=5,in=270]($(4)+(0,3)$)
	to[out=90,in=0]($(5)$)
	to[out=0,in=45]($(4)+(-2,3)$)
	to[out=225,in=0]($(2)$)
	to[out=0,in=90]($(4)+(-2,0)$)
	to[out=270,in=15]($(1)$)
	;
	\end{scope}
	\end{tikzpicture}}
} \qquad \qquad\qquad 
\subfigure[Embedding in $\bR^3$]{
		\begin{tikzpicture}[thick,scale=6.5]
	\coordinate [label=above:1](A1) at (0.3,0.5);
	\coordinate [label=left:2] (A2) at (0,0);
	\coordinate  [label=right:3] (A3) at (0.6,0);
	\coordinate [label=right:4] (A4) at (0.6,0.2);
	\coordinate [label=below:5] (A5) at (0.2,-0.3);
	\coordinate (A6) at (0.334,-0.062);
	
	\begin{scope}[thick,dashed,,opacity=0.6]
	\draw [fill=blue, opacity=0.3] (A1) -- (A2) -- (A4) -- cycle;
	\draw [fill=black, opacity = 0.4](A1) -- (A3) -- (A4) -- cycle;
	\draw [fill=blue, opacity = 0.3](A1) -- (A2) -- (A3) -- cycle;
	\draw [pattern= north west lines, pattern color=blue] (A2) -- (A4) -- (A5) -- cycle;
	\draw [pattern= north west lines, pattern color=black](A1) -- (A2) -- (A5) -- cycle;
	\draw [pattern= north east lines, pattern color=green] (A2) -- (A3) -- (A5) -- cycle;
	\end{scope}
	\draw (A3) -- (A5);
	\draw (A2) -- (A5);
	\draw (A2) -- (A3);
	\draw (A1) -- (A2);
	\draw (A1) -- (A3);
	\draw (A1) -- (A4);
	\draw (A3) -- (A4);
	\end{tikzpicture}
}\caption{Example of a minimal topologically $3$-connected hypergraph on $5$ vertices with 6 hyperedges $\{ \{1,2,3\}, \{1,2,4\}, \{1,3,4\}, \{1,2,5\}, \{2,3,5\}, \{2,4,5\} \}$.} \label{fig.connected_3hypergraph}
\end{figure*}

The main property for minimally topologically $k$-connected hypergraphs is summarized in the following lemma. We remark that this property is implicitly implied by the main theorem in \cite{Kalai1983}. 

\begin{lemma}\label{thm.connectivity}
Any minimal topological $k$-connected hypergraph with $n$ vertices has exactly $\binom{n-1}{k-1}$ hyperedges. 
\end{lemma}

\nb{A detailed proof of Lemma~\ref{thm.connectivity} is in Appendix~\ref{lemma1}.} When $k=2$, Lemma \ref{thm.connectivity} generalizes the fact that a tree on $n$ vertices has exactly $n-1$ edges. The topological interpretation of Lemma \ref{thm.connectivity} is as follows: embed the hypergraph into $\reals^k$ and think of hyperedges as faces (as in Figure \ref{fig.connected_3hypergraph} as an example). For a minimal topologically $k$-connected hypergraph, the minimality ensures that the facets cannot be the boundary of a closed domain. As a result, these facets can be shrunk into a single point topologically, which is of Euler characteristic $1$. Moreover, for $1\le j\le k-1$, let $F_j$ be the number of $(j-1)$-dimensional edges, the topological connectivity condition ensures that $F_j=\binom{n}{j}$. Now by Euler's formula \cite{rotman1998topology}, the number $F$ of faces equals to
\begin{align*}
F &= (-1)^{k-1} \left( 1 - \sum_{j=1}^{k-1} (-1)^{j-1}F_j \right) \\
&= \sum_{j=0}^{k-1} (-1)^{k-1-j}\binom{n}{j} = \binom{n-1}{k-1},
\end{align*} 
confirming Lemma \ref{thm.connectivity}. 

\subsection{Achievability scheme}\label{subsec.strategy}
In this subsection we propose the achievability scheme for general topologically $k$-connected hypergraph $G$. \nb{Without loss of generality we assume that $G$ is minimal topologically $k$-connected, for we can always ignore the other edges and consider a minimal topologically connected subgraph.} For each $i\in [n]$, we define the induced hypergraph $G_i$ from $G$ as follows: the vertex set of $G_i$ is $V_i = [n]\backslash \{i\}$, and the edge set of $G_i$ is $E_i = \{e\backslash \{i\}: i\in e\in E \}$. Hence, the induced hypergraph $G_i$ is $(k-1)$-uniform, and $e$ is a hyperedge of $G_i$ if and only if $e\cup \{i\}\in E$. We have the following lemma. 
\begin{lemma}\label{lemma.induced_hypergraph}
For $k\ge 3$, if $G$ is topologically $k$-connected, then all induced hypergraphs $G_i$ are topologically $(k-1)$-connected. 
\end{lemma}

\nb{A detailed proof of Lemma~\ref{lemma.induced_hypergraph} is in Appendix~\ref{lemma2}.} We propose the following communication strategy for topologically $k$-connected hypergraphs. For each edge $e\in E$, we define an independent random variable $R_e\sim\mathsf{Unif}(\{0,1\})$ by tossing the associated common coin.
\begin{definition}[Communication strategy for $k$-connected hypergraphs]\label{def.strategy}
For a minimal topologically $k$-connected hypergraph $G$ with $k\ge 3$, the communication strategy is as follows: for each $i\in [n]$, 
\begin{enumerate}
	\item Player $i$ constructs the induced hypergraph $G_i$, and choose an arbitrary minimal topologically $(k-1)$-connected subgraph $G_i^\star\subseteq G_i$ (existence of $G_i^\star$ is ensured by Lemma \ref{lemma.induced_hypergraph}); 
	\item For each hyperedge $e$ of $G_i$ which is not in $G_i^\star$, let $e$ be generated by $e_1,\cdots,e_m$ in $G_i^\star$ (cf. Definition \ref{def.connectivity}). Player $i$ then writes $R_{e\cup \{i\}} \oplus R_{e_1\cup \{i\}}\oplus \cdots \oplus R_{e_m\cup \{i\}}$ on the blackboard. 
\end{enumerate}
\end{definition}

Although the previous scheme is defined for $k\ge 3$, it is straightforward to see that it reduces exactly to the achievability scheme in Appendix \ref{subsec.k=2_general} when $k=2$ (by adapting the definition of topologically $1$-connected graph appropriately). Moreover, this strategy can be implemented under the simultaneous message passing model. We refer to Figure \ref{fig.strategy_example} for an example. 

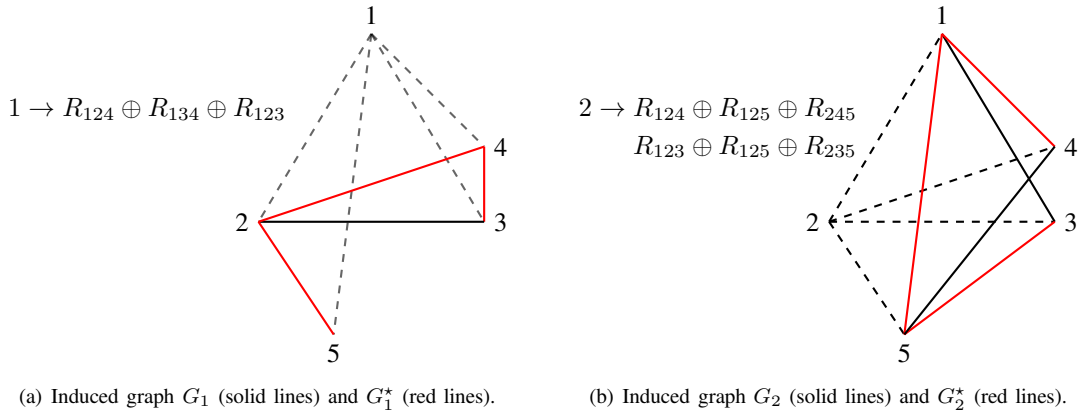
\begin{figure}[h]
	\subfigure[Induced graph $G_1$ (solid lines) and $G_1^\star$ (red lines).]{
		\centering
		\begin{tikzpicture}[thick,scale=5]
		\coordinate [label=above:1](A1) at (0.3,0.5);
		\coordinate [label=left:2] (A2) at (0,0);
		\coordinate  [label=right:3] (A3) at (0.6,0);
		\coordinate [label=right:4] (A4) at (0.6,0.2);
		\coordinate [label=below:5] (A5) at (0.2,-0.3);
		\coordinate [label=left: $1 \rightarrow R_{124}\oplus R_{134}\oplus R_{123}$] (A6) at (0.1,0.3);
		
		\begin{scope}[thick,dashed,,opacity=0.6]
		\draw  (A1) -- (A2);
		\draw  (A1) -- (A3);
		\draw  (A1) -- (A4);
		\draw  (A1) -- (A5);
		\end{scope}
		\draw [red] (A3) -- (A4);
		\draw (A2) -- (A3);
		\draw [red] (A2) -- (A5);
		\draw [red] (A2) -- (A4);
		
		\end{tikzpicture}
	}
	\quad 
	\subfigure[Induced graph $G_2$ (solid lines) and $G_2^\star$ (red lines).]{
		\centering
		\begin{tikzpicture}[thick,scale=5]
		\coordinate [label=above:1](A1) at (0.3,1.5);
		\coordinate [label=left:2] (A2) at (0,1);
		\coordinate  [label=right:3] (A3) at (0.6,1);
		\coordinate [label=right:4] (A4) at (0.6,1.2);
		\coordinate [label=below:5] (A5) at (0.2,0.7);
		\coordinate [label=left: $2 \rightarrow R_{124}\oplus R_{125}\oplus R_{245}$] (A6) at (0.1,1.3);
		\coordinate [label=left: $ R_{123}\oplus R_{125}\oplus R_{235}$] (A7) at (0.1,1.2);
		
		\begin{scope}[thick,dashed]
		\draw  (A2) -- (A1);
		\draw  (A2) -- (A3);
		\draw  (A2) -- (A4);
		\draw  (A2) -- (A5);
		\end{scope}
		\draw (A1) -- (A3);
		\draw [red] (A1) -- (A5);
		\draw [red] (A1) -- (A4);
		\draw [red] (A3) -- (A5);
		\draw  (A4) -- (A5);
		
		\end{tikzpicture}
	}
	\caption{The communication strategy on the minimally topologically connected 3-uniform hypergraph in Figure \ref{fig.connected_3hypergraph}, which achieves the optimal communication rate $1/2$. }\label{fig.strategy_example}
\end{figure}

Assuming for a moment that every player may decode the random vector $X=(R_e: e\in E)$, we show that the communication rate of this strategy is optimal. Firstly, by Lemma \ref{thm.connectivity} and the minimality of $G$, 
$
H(X) = |E| = \binom{n-1}{k-1}. 
$
Moreover, the number of bits player $i$ writes on the blackboard is
$
|M_i| = |\{e\in E: i\in e\}| - \binom{n-2}{k-2},
$
where Lemma \ref{thm.connectivity} again shows that each $G_i^\star$ has $\binom{n-2}{k-2}$ hyperedges. As a result, the total length of the message $M$ is
\begin{align*}
|M| &= \sum_{i=1}^n |M_i| = \sum_{i=1}^n \left(|\{e\in E: i\in e\}| - \binom{n-2}{k-2}\right)\\
& = k|E| - n\binom{n-2}{k-2} = \binom{n-2}{k-1}. 
\end{align*}
Hence, the communication rate can be upper bounded as
\begin{align*}
\frac{H(M)}{H(X)} \le \frac{|M|}{H(X)} = \frac{\binom{n-2}{k-1}}{\binom{n-1}{k-1}} = \frac{n-k}{n-1}, 
\end{align*}
which is optimal by Corollary \ref{cor.lower}. Therefore it remains to prove the following theorem. 
\begin{theorem}\label{thm.decodable}
	Let $G=(V,E)$ be a topologically $k$-connected hypergraph. Then under the communication strategy in Definition \ref{def.strategy}, every player may decode the random vector $X$. 
\end{theorem}

The proof of Theorem \ref{thm.decodable} requires delicate algebraic and combinatorial arguments for topological connectivity, which is deferred to Appendix \ref{subsec.correctness}. 

\section{Generalization: Clusters of Connected Components}\label{sec.generalization}
In this section, we generalize the achievability scheme in Section \ref{sec.achievability} to incorporate the cases where the hypergraph is not topologically connected but consists of topologically connected components. 

\subsection{Path connectivity}
First we review the notion of path connectivity in general (and not necessarily uniform) hypergraphs. Recall that a general hypergraph $G=(V,E)$ consists of a finite vertex set $V$ and a finite hyperedge set $E=\{A_1,\cdots,A_m\}$, where $A_i\subseteq V$ are non-empty subsets of $V$. Path connectivity in hypergraphs is defined as follows. 
\begin{definition}[Path and path connectivity]\label{def.path_connectivity}
	In a hypergraph $G=(V,E)$ and any vertices $u,v\in V$, a \emph{simple path} from $u$ to $v$ is a sequence of distinct vertices $v_0,v_1,\cdots,v_k\in V$ and distinct hyperedges $A_1,\cdots,A_k\in E$ such that $v_0=u, v_k=v$, and $v_{i-1},v_i\in A_i$ for any $i\in [k]$. The hypergraph $G$ is \emph{path-connected} iff for any $u,v\in V$, there is a simple path from $u$ to $v$.  
\end{definition}

We also need the notion of cycle-free hypergraphs as follows. 
\begin{definition}[Simple cycle and cycle-free hypergraph]\label{def.cycle_free}
	In a hypergraph $G=(V,E)$, a \emph{simple cycle} is a sequence of distinct vertices $v_0,v_1,\cdots,v_{k-1}\in V$ and distinct hyperedges $A_1,\cdots,A_k\in E$ such that $v_{i-1}, v_i\in A_i$ for any $i\in [k]$, where $v_k = v_0$. The hypergraph $G$ is \emph{cycle-free} iff there is no simple cycle in $G$. 
\end{definition}

Note that a path-connected cycle-free $2$-uniform hypergraph is a tree. The next lemma is another generalization of the fact that a tree on $n$ vertices has exactly $n-1$ edges. Recall that for each $v\in V$, the degree of $v$ is defined as $\deg(v) = |\{A \in E: v\in A  \}|$.
\begin{lemma}\label{lemma.path_connectivity_edge}
	Let $G=(V,E)$ be a path-connected cycle-free hypergraph. Then
$
	\sum_{A \in E} \left(|A| - 1 \right) = |V| - 1,
$ and $
\sum_{v\in V} (\deg(v) - 1) = |E| - 1.
$
\end{lemma}
\nb{A detailed proof of Lemma~\ref{lemma.path_connectivity_edge} is in Appendix~\ref{lemma3}.}
\subsection{Achievability scheme}
In this section we formally define the cluster of connected components, and present a communication strategy achieving the upper bound in Theorem \ref{thm.main.upper} under the simultaneous message passing procotol. 
\begin{definition}\label{def.cluster}
	Let $G=(V,E)$ be a $k$-uniform hypergraph. We call $G$ is a \emph{cluster of connected components} if and only if there is another hypergraph (not necessarily $k$-uniform) $G_c = (V, \{A_1,\cdots,A_m\})$ such that (where the subscript $c$ stands for ``cluster''): 
	\begin{enumerate}
		\item the hypergraph $G_c$ is path-connected and cycle-free; 
		\item for each $i\in [m]$, the restriction of $G$ on the vertices in $A_i$ is topologically $k$-connected. 
	\end{enumerate}
\end{definition}
\begin{figure}[h]
	\centering
	\begin{tikzpicture}[thick,scale=5]
	\coordinate [label=above:1](A1) at (0.3,0.5);
	\coordinate [label=left:4] (A2) at (0.2,0);
	\coordinate  [label=right:5] (A3) at (0.6,0);
	\coordinate [label=right:6] (A4) at (0.6,0.2);
	\coordinate [label=below:2] (A5) at (-0.1,0.1);
	\coordinate [label=left:3] (A6) at (-0.2,0.25);
	
	\begin{scope}[thick,dashed,,opacity=0.6]
	\draw [fill=blue, opacity=0.3] (A1) -- (A2) -- (A4) -- cycle;
	\draw [fill=blue, opacity = 0.5](A2) -- (A3) -- (A4) -- cycle;
	\draw [fill=blue, opacity = 0.3](A1) -- (A2) -- (A3) -- cycle;
	\draw [fill=green, opacity = 0.3](A1) -- (A5) -- (A6) -- cycle;
	
	\end{scope}
	\draw (A1) -- (A6);
	\draw (A1) -- (A5);
	\draw (A6) -- (A5);
	\draw (A2) -- (A3);
	\draw (A1) -- (A2);
	\draw (A2) -- (A4);
	\end{tikzpicture}
	\caption{An example of a cluster of connected components.} \label{fig.cluster}
\end{figure}
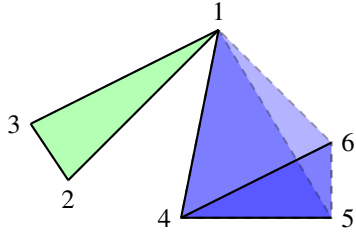

Definition \ref{def.cluster} essentially says that to form a cluster, the topologically $k$-connected components of $G$ should be path-connected without cycles in terms of components. Figure \ref{fig.cluster} illustrates an example of such a cluster, where 
\begin{align*}
G&=([6], \{ \{1,2,3\}, \{1,4,5\}, \{1,4,6\}, \{4,5,6\} \}), \\
G_c &= ([6], \{ \{1,2,3\}, \{1,4,5,6\} \}). 
\end{align*}

Next we define the communication strategy for clusters of connected components. 
\begin{definition}[Communication strategy for clusters of connected components]\label{def.strategy_cluster}
	Let the $k$-uniform hypergraph $G=(V,E)$ be a cluster of connected components, with the corresponding cluster hypergraph $G_c=(V,\{A_1,\cdots,A_m\})$. The communication strategy is as follows: 
	\begin{enumerate}
		\item For each $i\in [m]$, remove hyperedges properly so that the restriction of $G$ on $A_i$ is minimally topologically $k$-connected; 
		\item Messages within components: for each $i\in [m]$, repeat (for different realizations of coin tosses) the strategy in Definition \ref{def.strategy} for $M_i$ times in the restricted graph on $A_i$, where $M_i$ is chosen so that
		\begin{align}\label{eq.choice_Mi}
		M_i \cdot \binom{|A_i| - 2}{k-2} = C
		\end{align}
		for some common constant $C>0$. We choose $C$ large enough so that each $M_i$ is an integer; 
		\item Messages across components: for each $v\in V$ belonging to at least two connected components $A_{i_1},\cdots,A_{i_\ell}$ (i.e., $\ell=\deg_{G_c}(v)\ge 2$) and $j\in [\ell]$, let $G_j^\star$ be the minimal topologically $(k-1)$-connected subgraph of $v$-induced hypergraph in the connected component $A_{i_j}$ (cf. Definition \ref{def.strategy}) used in the previous step. Let $R_j \in \FF_2^C$ be the binary vector consisting of the outcomes of coin tosses corresponding to every hyperedge in $G_j^\star$ repeated $M_{i_j}$ times\footnote{Note that $G_j^\star$ has exactly $\binom{|A_{i_j}-2|}{k-2}$ hyperedges by Lemma \ref{thm.connectivity}, the choice of $M_{i_j}$ in \eqref{eq.choice_Mi} ensures that the dimension of the vector $R_j$ is exactly $C$.}, in an arbitrary order. Then the vertex $v$ writes
		\begin{align*}
		M_v = (R_1 \oplus R_2, R_1\oplus R_3, \cdots, R_1\oplus R_\ell)
		\end{align*} 
		on the blackboard. 
	\end{enumerate}
\end{definition}

The intuition behind the strategy in Definition \ref{def.strategy_cluster} is as follows. Firstly, each connected component employs the strategy in Definition \ref{def.strategy} so that each vertex in this component may decode all coin tossing outcomes within that component. Secondly, for vertices which link multiple connected components, they employ the strategy in Appendix~\ref{subsec.k=2_general} to share coin tossing outcomes from different components. Finally, since different connected components may be of different sizes, proper repetitions are necessary to ensure that all components have the same amount of information to be shared across components. 

For example, for the previous hypergraph in Figure \ref{fig.cluster}, we have $|A_1|=3, |A_2|=4$. Consequently, we may choose $M_1= 2, M_2=1$ and $C=2$. Let $R_{123}, R_{123}'$ be independent outcomes of the common coin shared among $\{1,2,3\}$ (i.e., toss coin twice), then the message within components (broadcast by player $4$) is
$
R_{145} \oplus R_{146} \oplus R_{456}, 
$
and the messages across components (broadcast by player $1$) are
$
R_{123} \oplus R_{145}, R_{123}' \oplus R_{146}. 
$
It is straightforward to see that each player may decode the random vector $(R_{123}, R_{123}', R_{145}, R_{146}, R_{456})$, and thus the previous strategy achieves the optimal communication rate $3/5$ in this example. 

The following theorem states that for general clusters of connected components, the strategy in Definition \ref{def.strategy_cluster} achieves the optimal communication rate. Let $X$ be the binary vector consisting of all coin tossing outcomes during the strategy in Definition \ref{def.strategy_cluster}.
\begin{theorem}\label{thm.cluster}
	For any $k$-uniform hypergraph $G=(V,E)$ which is a path-connected cycle-free cluster of topologically connected components (cf. Definition \ref{def.cluster}), every player may decode the entire outcome vector $X$ under the strategy in Definition \ref{def.strategy_cluster}, with communication rate $H(M)/H(X) = (n-k)/(n-1)$.
\end{theorem}
\nb{A detailed proof of Theorem~\ref{thm.cluster} is in Appendix~\ref{theorem4}.}
\subsection{Further discussions on star graphs}\label{subsec.non-examples}
Motivated by Theorem \ref{thm.cluster}, a natural question arises on whether any $k$-uniform hypergraph which is possible to achieve the optimal communication rate $(n-k)/(n-1)$ must contain a path-connected cycle-free cluster of topologically connected components. For $k=2$, examples in Appendix \ref{sec.example} show that the answer is affirmative. However, in this section we show that even for $k=3$ a richer class of hypergraphs achieves the optimal communication rate. Also, for the special case of \nb{star graphs (a hypergraph with a single vertex contained in all hyperedges)}, we characterize a necessary and sufficient condition for any $3$-uniform star graph to achieve the optimal communication rate. Hence, it is an outstanding open problem to characterize the entire class of communication-optimal hypergraphs.

We first construct an example of a $3$-uniform hypergraph not satisfying the assumption of Theorem \ref{thm.cluster} but achieves the optimal communication rate of $(n-k)/(n-1)$. The graph $G$ is shown in Figure \ref{fig.non-examples-1}, with
$$G = ([6], \{ \{1,2,3\}, \{1,3,4\}, \{1,4,5\}, \{1,5,6\}, \{1,2,6\} \}).$$
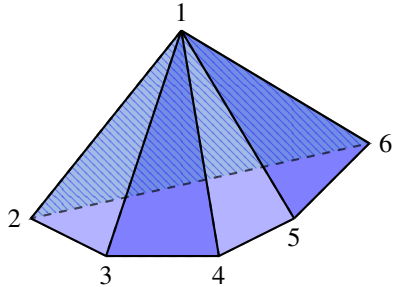
\begin{figure}[h]
\centering{
\begin{tikzpicture}[thick,scale=5]
		\coordinate [label=above:1](A1) at (0.3,0.6);
		\coordinate [label=left:2] (A2) at (-0.1,0.1);
		\coordinate  [label=below:3] (A3) at (0.1,0);
		\coordinate [label=below:4] (A4) at (0.4,0);
		\coordinate [label=below:5] (A5) at (0.6,0.1);
		\coordinate [label=right:6] (A6) at (0.8,0.3);
		
		\begin{scope}[thick,dashed,,opacity=0.6]
		\draw [fill=blue, opacity=0.3] (A1) -- (A2) -- (A3) -- cycle;
		\draw [fill=blue, opacity = 0.5](A1) -- (A3) -- (A4) -- cycle;
		\draw [fill=blue, opacity = 0.3](A1) -- (A4) -- (A5) -- cycle;
		\draw [fill=blue, opacity = 0.5](A1) -- (A5) -- (A6) -- cycle;
		\draw [fill=green, opacity = 0.1](A1) -- (A2) -- (A6) -- cycle;
\draw [pattern= north west lines, pattern color=blue] (A1) -- (A2) -- (A6) -- cycle;		
		\end{scope}
		\draw (A1) -- (A2);
		\draw (A1) -- (A3);
		\draw (A1) -- (A4);
		\draw (A1) -- (A5);
		\draw (A1) -- (A6);
		\draw (A2) -- (A3);
		\draw (A3) -- (A4);
		\draw (A4) -- (A5);
		\draw (A5) -- (A6);
		
\end{tikzpicture}}
\caption{An example hypergraph not satisfying the condition of Theorem \ref{thm.cluster}.}\label{fig.non-examples-1}
\end{figure}

It is not hard to show that $G$ is not a path-connected cycle-free cluster of topologically connected components, as the only topologically connected components are the single triangles and the resulting hypergraph $G_c$ will not be cycle-free. Hence, $G$ does not satisfy the condition of Theorem \ref{thm.cluster}. However, the optimal communication rate $3/5$ can be achieved for $G$, where a feasible strategy is that player $1$ writes the following message $M$ on the blackboard:
$$ M = (R_{123} \oplus R_{145}, R_{134}\oplus R_{156}, R_{123}\oplus R_{134} \oplus R_{126}).$$
One can easily verify that given the $3$-bit message $M$, each player is able to decode the entire $5$-bit randomness. 

The above example is a special case of a $3$-uniform \emph{star graph}, i.e. a $3$-uniform hypergraph where every edge contains a common vertex $v^\star$. In fact, the above strategy can be generalized for general star graphs, and the following theorem completely characterizes the family of $3$-uniform star graphs where the optimal communication rate $(n-3)/(n-1)$ is achievable. 

\begin{theorem}\label{thm.star-graph}
Let $G$ be a $3$-uniform star graph with $n$ vertices and the central vertex $v^\star$, and $G_{v^\star}$ be the induced graph (which is a classical graph) at vertex $v^\star$ as per Section \ref{subsec.strategy}. Then the optimal communication rate $(n-3)/(n-1)$ can be achieved for $G$ if and only if $G_{v^\star}$ contains a vertex-disjoint union of simple edges or Hamilton cycles of odd length including all vertices. 
\end{theorem}

For example, the induced graph $G_1$ for the hypergraph $G$ in Figure \ref{fig.non-examples-1} is a Hamilton cycle on all vertices $\{2,3,\cdots,6\}$, and therefore satisfies the condition of Theorem \ref{thm.star-graph}. The \emph{if} part of Theorem \ref{thm.star-graph} is shown by providing an explicit communication strategy in same spirits to the above example, and the \emph{only if} part is more challenging and requires the theory of fractional graphs. The complete proof is presented in Appendix \ref{appendix.star_graphs}.

\appendices
\section{Simple Examples}\label{sec.example}
In this section we provide some examples where the hypergraph $G=(V,E)$ is rather simple, and propose the corresponding achievability schemes. 
\subsection{Star graph with $k=2$}\label{subsec.k=2_easy}
In the star graph case with $k=2$, there are $n\ge 3$ players where the last player shares a common fair coin with any other player (i.e., the associated graph $G$ is a star graph with center vertex $n$). First consider $n=3$, and let $R_i, i\in \{1,2\}$ be the outcome (head or tail) of the first toss of the common coin shared between player $i$ and $3$. Clearly $R_1$ and $R_2$ are independent $\mathsf{Unif}(\{0,1\})$ random variables, and we consider the strategy that player $3$ writes $M=R_1\oplus R_2$ on the blackboard (cf. Figure \ref{fig.star_graph_n=3}). Since $R_2=R_1\oplus M$ and $R_1 = R_2\oplus M$, all players may know $R_1, R_2$ perfectly and generate $X=(R_1, R_2)$. Note that
\begin{align*}
H(X) = 2, \qquad H(M) = 1,
\end{align*}
we have achieved the optimal communication rate $\frac{1}{2}$, confirming Theorem \ref{thm.main.upper}. 

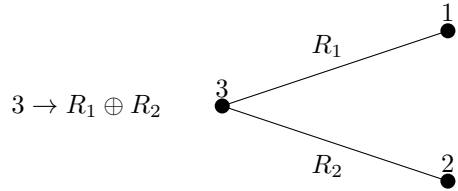
\begin{figure}[h]
	\centering
	\scalebox{0.77}{
		\begin{tikzpicture}
		\node (1) at (4,11) {};
		\node (2) at (4,9) {};
		\node (3) at (1,10) {};
		
		\node at (2.4,10.8) {$R_1$};
		\node at (2.4,9.2) {$R_2$};
		\node at (-0.8,10) {$3 \rightarrow R_1\oplus R_2$};
		
		\fill (1) circle (0.1) node [above] {$1$};
		\fill (2) circle (0.1) node [above] {$2$};
		\fill (3) circle (0.1) node [above] {$3$};
		
		\begin{scope}
		\draw($(1)$) to ($(3)$);
		\draw($(2)$) to ($(3)$);
		\end{scope}
		\end{tikzpicture}
	}
	\caption{Communication strategy for star graph with $n=3$, $k=2$.}\label{fig.star_graph_n=3}
\end{figure}

The achievability scheme for $n\ge 3$ is similar. Let $R_i, 1\le i\le n-1$ be independent $\mathsf{Unif}(\{0,1\})$ random variables shared between player $i$ and $n$, consider the case where the last player broadcasts the following message on the blackboard: 
\begin{align*}
M = (R_1\oplus R_2, R_1\oplus R_3, \cdots, R_1\oplus R_{n-1}).
\end{align*}
Based on the message $M$, player $1$ may decode any other $R_i$ using the knowledge of $R_1$. For any player $j\in \{2,\cdots,n-1\}$, knowing both $R_1\oplus R_j$ from $M$ and $R_j$, player $j$ can decode $R_1$ and further all $R_i$ based on $M$. Hence, in this case all player may generate $X=(R_1,\cdots,R_{n-1})$, with
\begin{align*}
H(X) = n - 1, \qquad H(M) = n-2, 
\end{align*}
achieving the optimal communication rate $\frac{n-2}{n-1}$. 

\subsection{General connected graph with $k=2$}\label{subsec.k=2_general}
We may generalize the strategy in Appendix~\ref{subsec.k=2_easy} to the case where $k=2$ and the graph $G$ is connected. For each edge $e\in E$, we may associate an independent random variable $R_e\sim\mathsf{Unif}(\{0,1\})$ by tossing the associated common coin. Since $G$ is connected, it contains a spanning tree $T\subseteq G$. Now consider the following strategy: for each player $i\in [n]$, 
\begin{enumerate}
	\item if the degree of $i$ in $T$ is 1, player $i$ writes nothing on the blackboard (i.e., $M_i=\emptyset$); 
	\item if the degree of $i$ in $T$ is at least 2, let $e_1,\cdots,e_{m_i}$ be all of its neighboring edges in an arbitrary order, with $m_i=\deg_T(i)$. player $i$ then writes $M_i=(R_{e_1}\oplus R_{e_2}, R_{e_1}\oplus R_{e_3}, \cdots, R_{e_1}\oplus R_{e_{m_i}} )$ on the blackboard. 
\end{enumerate}

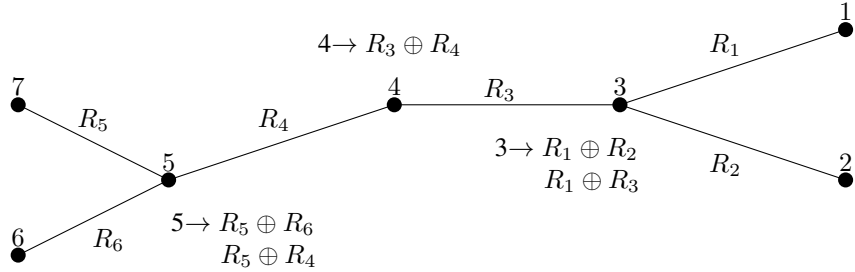
\begin{figure}[h]
	\centering
	\scalebox{0.77}{
		\begin{tikzpicture}
		\node (1) at (4,11) {};
		\node (2) at (4,9) {};
		\node (3) at (1,10) {};
		\node (4) at (-2,10) {};
		\node (5) at (-5,9) {};
		\node (6) at (-7,8) {};
		\node (7) at (-7,10) {};

		\node at (2.4,10.8) {$R_1$};
		\node at (2.4,9.2) {$R_2$};
		\node at (-0.6,10.2) {$R_3$};
		\node at (-3.6,9.8) {$R_4$};
		\node at (-6,9.8) {$R_5$};
		\node at (-5.8,8.2) {$R_6$};
		
		\node[align=left] at (0.3,9.2) {3$\rightarrow R_1\oplus R_2$ \\ \phantom{3$\rightarrow $ }$R_1\oplus R_3$ };
		\node[align=left] at (-2,10.8) {4$\rightarrow R_3\oplus R_4$ };
		\node[align=left] at (-4,8.2) {5$\rightarrow R_5\oplus R_6$ \\ \phantom{5$\rightarrow $ }$R_5\oplus R_4$ };

		\fill (1) circle (0.1) node [above] {$1$};
		\fill (2) circle (0.1) node [above] {$2$};
		\fill (3) circle (0.1) node [above] {$3$};
		\fill (4) circle (0.1) node [above] {$4$};
		\fill (5) circle (0.1) node [above] {$5$};
		\fill (6) circle (0.1) node [above] {$6$};
		\fill (7) circle (0.1) node [above] {$7$};
		
		\begin{scope}
		\draw($(1)$) to ($(3)$);
		\draw($(2)$) to ($(3)$);
		\draw($(3)$) to ($(4)$);
		\draw($(4)$) to ($(5)$);
		\draw($(5)$) to ($(6)$);
		\draw($(5)$) to ($(7)$);
		
		\end{scope}
		\end{tikzpicture}
	}
	\caption{Communication strategy for a tree with $n=7$, $k=2$.}\label{fig.connected_k=2}
\end{figure}

An example of this strategy is illustrated in Figure \ref{fig.connected_k=2}. The next lemma shows that every player may generate the random vector $X=(R_e: e\in E_T)$, where $E_T$ is the edge set of the spanning tree $T$. 
\begin{lemma}\label{lemma.k=2_connected}
	Based on the message $M=(M_1,\cdots,M_n)$, every player can decode $X=(R_e: e\in E_T)$. 
\end{lemma}
\begin{proof}
	By symmetry, it suffices to prove that the first player can decode $X$. We prove the following statement: for any edge $(i,j)\in E_T$, if player $1$ can decode $(R_e: i\in e\in E_T)$, then he can also decode $(R_e: j\in e\in E_T)$. The proof of this statement exactly follows from the arguments in Appendix~\ref{subsec.k=2_easy} based on the star graph centered at $i$ and the message $M_i$. Now since $T$ is connected, we may start from $i=1$ in the previous statement and visit all vertices of $T$, completing the proof. 
\end{proof}

Next we evaluate $H(X)$ and $H(M)$. Clearly 
\begin{align*}
H(X) &= |E_T| = n-1, \\
H(M) &\le |M| = \sum_{i=1}^n (\deg_T(i) - 1) = 2|E_T| - n = n-2. 
\end{align*}
As a result, $H(M)\le \frac{n-2}{n-1}H(X)$, proving Theorem \ref{thm.main.upper} for the case $k=2$. 

\subsection{Forehead model with $k=n-1$}\label{subsec.forehead}
In the forehead model, we have $k=n-1$, and $G$ is a complete $k$-uniform hypergraph. As usual, for each $i\in [n]$, we associate an independent random variable $R_{\backslash i}\sim \mathsf{Unif}(\{0,1\})$ via coin tossing, and player $i$ knows all random variables except $R_{\backslash i}$. This is where the name \emph{forehead model} comes from: the random variable $R_{\backslash i}$ is written on the forehead of player $i$ which he cannot see \cite{Chandra:1983}. The communication strategy for this model is as follows: player $1$ writes 
\begin{align*}
M = R_{\backslash 2} \oplus R_{\backslash 3} \oplus \cdots \oplus R_{\backslash n}
\end{align*}
on the blackboard, and other players write nothing. It is clear that everyone then may know and generate $X=(R_{\backslash 2}, R_{\backslash 3}, \cdots, R_{\backslash n})$, with
\begin{align*}
H(X) = n-1, \qquad H(M) = 1. 
\end{align*}
Hence, this strategy provides an achievability scheme of $H(M)=\frac{1}{n-1}H(X)$ in the forehead model, conforming to Theorem \ref{thm.main.upper}. 

\section{Asymptotically Optimal Communication Rates}\label{sec.lower}
This section is devoted to the asymptotically optimal communication rates for common randomness generation. Specifically, we first prove the lower bounds in Theorem \ref{thm.main.lower} and Corollary \ref{cor.lower}, and then show that the rate given by the linear programming is attainable asymptotically. 

\subsection{Proof of Theorem \ref{thm.main.lower}}\label{theorem1}
We start with some notations. Recall that $X$ is the outputted common randomness, and $M$ is the message written on the blackboard. Fix any complete order relationship $(E,<)$ on the edge set $E$, and for $e\in E$, let $R_e$ be the randomness associated with edge $e$, and $R_{<e}$ be the set of randomness associated with edges preceding $e$ under the order $(E,<)$. Furthermore, for any $U\subseteq V$ we denote by $R_U$ the set of randomness known to the player set $U$.  

By scaling, it suffices to find non-negative parameters $(r_v)_{v\in V}, (s_e)_{e\in E}$ such that the following inequalities hold: 
\begin{align}
\sum_{v\in U} r_v &\ge \sum_{e\in E: e\subseteq U} s_e, \qquad \forall U\subsetneq V \label{eq.cut_condition} \\
\sum_{v\in V} r_v &\le H(M), \label{eq.H_M} \\
\sum_{e\in E} s_e &\ge H(X). \label{eq.H_X}
\end{align}
Intuitively, the quantity $r_v$ denotes the length of the messages sent by player $v$, and $s_e$ denotes the number of bits in $R_e$ used to generate the common output $X$. To specify the choices, recall that a blackboard communication protocol can be treated as an infinite-round sequential communication, and we write $M = (M_1,M_2,\cdots)$ where $M_t$ is outputted by the player $t\bmod n$ and may be an empty string. Now we set
\begin{align*}
r_v &= \sum_{t=0}^\infty H(M_{tn + v} | M^{tn+v-1}), \qquad \forall v\in V = [n], \\
s_e &= I(X; R_e | R_{<e}), \qquad \forall e\in E. 
\end{align*}

We verify the inequalities \eqref{eq.cut_condition}--\eqref{eq.H_X}. To establish \eqref{eq.cut_condition}, note that
\begin{align*}
\sum_{v\in U} r_v &=\sum_{t=0}^\infty  \sum_{v\in U} H(M_{tn + v} | M^{tn+v-1}) \\
&\ge \sum_{t=0}^\infty  \sum_{v\in U} H(M_{tn + v} | M^{tn+v-1}, R_{U^c}) \\
&\stepa{=} \sum_{t=0}^\infty  \sum_{v\in V} H(M_{tn + v} | M^{tn+v-1}, R_{U^c}) \\
&\stepb{=} H(M|R_{U^c}) \\
&\stepc{\ge } H(X|R_{U^c}) \\
&\stepd{=} H(X|R_{U^c}) - H(X|R_{U^c}, (R_e)_{e\subseteq U} ) \\
&= I(X; (R_e)_{e\subseteq U} | R_{U^c}) \\
&= \sum_{e\in E: e\subseteq U} I(X;R_e | R_{U^c}, (R_{e'})_{e'\in U, e'<e}) \\
&\stepe{\ge} \sum_{e\in E: e\subseteq U} I(X;R_e | R_{<e}) \\
&= \sum_{e\in E: e\subseteq U} s_e,
\end{align*}
where (a) follows from the fact that under the blackboard communication protocol $M_{tn+v}$ must be a function of $(M^{tn+v-1}, R_{U^c})$ whenever $v\in U^c$, (b) is due to the chain rule of the Shannon entropy, (c) is due to that $X$ is a function of $(M,R_{U^c})$ since each player $v\in U^c$ can output $X$ based on the message $M$ and her known randomness, (d) is due to that the output $X$ is a function of all randomness $(R_e)_{e\in E}$, and (e) follows from the inequality $I(A;B|C,D)\ge I(A;B|C)$ whenever $B$ and $D$ are conditionally independent given $C$. Therefore \eqref{eq.cut_condition} holds. The inequality \eqref{eq.H_M} holds with equality due to the chain rule of the Shannon entropy. For inequality \eqref{eq.H_X}, the chain rule gives
\begin{align*}
\sum_{e\in E} s_e = I(X; (R_e)_{e\in E}) = H(X)
\end{align*}
since the output $X$ is a function of $(R_e)_{e\in E}$. 

\nb{\subsection{Proof of Corollary \ref{cor.upper}}\label{corollary1proof}
We first show that $t(G)\le 1$ for all hypergraphs. Assigning non-negative weights $(s_e)_{e\in E}$ in an arbitrary way with $\sum_{e\in E} s_e = 1$, consider the following feasible solution $(r_v)_{v\in V}$: 
\begin{align*}
	r_v = \sum_{e\in E: v\in e} \frac{s_e}{|e|}, 
\end{align*}
where $|e|$ denotes the number of vertices in the hyperedge $e$. It is then clear that for all $U\subseteq V$, 
\begin{align*}
	\sum_{v\in U} r_v &= \sum_{v\in U} \sum_{e\in E: v\in e} \frac{s_e}{|e|} \\
	&= \sum_{e\in E} s_e\cdot \frac{\#\text{ of vertices }v\text{ in }U\text{ with }v\in e}{|e|} \\
	&\ge \sum_{e\in E: e\subseteq U} s_e,
\end{align*}
showing that $(r_v)_{v\in V}$ is indeed a feasible solution. Consequently, 
\begin{align*}
	t(G) \le \sum_{v\in V} r_v = \sum_{v\in V} \sum_{e\in E: v\in e} \frac{s_e}{|e|} = \sum_{e\in E} s_e = 1. 
\end{align*}

Next we prove that $t(G)=1$ if and only if $G$ is disconnected. For the \emph{if} part, for disconnected $G$, we may split the vertex set $V$ into two non-empty sets $U$ and $V\backslash U$, such that for every hyperedge $e$, either $e\subseteq U$ or $e\subseteq V\backslash U$. Consequently, for any feasible solution $(r_v)_{v\in V}$ and $(s_e)_{e\in E}$, 
\begin{align*}
	\sum_{v\in V} r_v &= \sum_{v\in U} r_v + \sum_{v\in V\backslash U} r_v \ge \sum_{e\in E: e\subseteq U} s_e + \sum_{e\in E: e\subseteq V\backslash U} s_e \\
	&= \sum_{e\in E}s_e \ge 1, 
\end{align*}
giving $t(G)\ge 1$. Since $t(G)\le 1$ for all hypergraphs, we have $t(G)=1$. 

For the \emph{only if} part, we prove the contrapositive that $t(G)<1$ if $G$ is connected. We construct a new graph $G'=(V,E')$ based on $G$: the new edge set $E'$ consists of all simple edges $(v,v')$ such that $\{v,v'\}\subseteq e$ for some hyperedge $e\in E$ (with multiplicities for each such $e$). We show that $t(G)\le t(G')$: in fact, for any feasible solution $(r_v')_{v\in V}$ and $(s_e')_{e\in E'}$ to the linear program for $G'$, the following solution
\begin{align*}
	r_v = r_v', \quad \forall v\in V,\qquad s_e = \sum_{e'\in E': e'\subseteq e} s_e', \quad \forall e\in E,
\end{align*}
is also feasible to the linear program for $G$, while with the same objective value. It remains to prove that $t(G')<1$. Since $G$ is connected, so is the $2$-uniform hypergraph $G'$. Now there are two ways to establish $t(G') < 1$. The first proof uses the operational meaning of $t(G')$, and it is shown in Appendix \ref{sec.example} that a communication rate $(n-2)/(n-1)$ could be achieved for any connected $2$-uniform graph $G'$. The second proof directly provides a feasible solution to the linear program for $G'$: find an arbitrary spanning tree $T=(V,E_T)$ of $G'$ with $|E_T| = n-1$, and set
\begin{align*}
	s_e = \frac{1}{n-1}\cdot \1(e\in E_T), \qquad r_v = \frac{\deg_T(v)-1 }{n-1}. 
\end{align*}
Clearly $\sum_{e\in E} s_e = 1$ and $\sum_{v\in V}r_v = (n-2)/(n-1)<1$. Now it suffices to check that this solution is feasible, i.e. for all non-empty $U\subsetneq V$, it holds that
\begin{align*}
	\sum_{v\in U} (\deg_T(v) - 1)  \ge \sum_{e\in E_T} \1(e\subseteq U). 
\end{align*}
Let $\text{cut}(U)$ be the cut size of $U$ in $T$, and $m(U)$ be the number of edges in the tree $T$ restricted to vertex set $U$. By simple algebra, the LHS is $2m(U) + \text{cut}(U) - |U|$, the RHS is $m(U)$, so it remains to show that $\text{cut}(U) + m(U) \ge |U|$. Since $T$ is a tree, it is clear that $\text{cut}(U) \ge C(U)$ and $m(U) = |U| - C(U)$, where $C(U)$ is the number of connected components in the restriction of $T$ to $U$; therefore, $\text{cut}(U) + m(U)\ge |U|$ holds. 
}

\subsection{Proof of Corollary \ref{cor.lower}}\label{corollary2proof}
Choosing $U = V \backslash \{v\}$ in Theorem \ref{thm.main.lower} for all $v\in V$ and summing up give
\begin{align*}
(n-1)\sum_{v\in V} r_v &= \sum_{v\in V}\sum_{u\in V\backslash \{v\}} r_u \\
&\stepa{\ge} \sum_{v\in V}\sum_{e\in E: e\subseteq V\backslash \{v\}} s_e \\
&\stepb{=} (n-k)\sum_{e\in E} s_e \\
&\stepc{\ge } n-k, 
\end{align*}
where inequalities (a) and (c) are due to the constraints in the linear program, and (b) follows from the fact that every edge $e$ is counted $n-k$ times in the summation in a $k$-uniform hypergraph. A rearrangement gives the proof. 

\subsection{An Asymptotic Achievability Scheme}\label{asymptachv}
The lower bound in Theorem \ref{thm.main.lower} is attainable asymptotically via linear network coding. The idea is essentially contained in \cite{nitinawarat2010perfect}, and we present it here for completeness. 

Let $t^\star$ be the minimum objective value of the linear program in Theorem \ref{thm.main.lower}. Then for any $t > t^\star$, there exists some feasible solution $(r_v)_{v\in V}, (s_e)_{e\in E}$ with $\sum_{v\in V}r_v / \sum_{e\in E} s_e \le t$ and all inequality constraints being strict. Let $N>0$ be a large integer, and without loss of generality we assume that $Nr_v, Ns_e$ are all integers. Consider the following scheme: 
\begin{enumerate}
	\item For any $e\in E$, toss the coin associated with the edge $e$ exactly $Ns_e$ times, and represent the outcomes by a binary vector $R_e\in \FF_2^{Ns_e}$; 
	\item For each player $v\in V$, she concatenates all vectors $R_e$ known to her into a long vector $z_v$ with length $\ell_v$, generates a random matrix $L_v$ uniformly distributed on $\FF_2^{Nr_v \times \ell_v}$, and writes the product $M_v = L_vz_v$ on the blackboard; 
	\item For decoding, each player $v\in V$ solves the linear system with observations $(z_v, (M_u)_{u\neq v} )$ to recover all vectors $(R_e)_{e\in E}$. 
\end{enumerate}

Clearly, the total length of the message written on the blackboard is $N\sum_{v\in V} r_v$, and the length of the output sequence is $N\sum_{e\in E} s_e$. Consequently, the communication rate is $\sum_{v\in V}r_v / \sum_{e\in E} s_e$ which is at most $t$. It remains to show that with positive probability, the above scheme is error free. Since the coding scheme is linear, a decoding error occurs iff there exists some non-zero vector $z=(z_v)_{v\in V}\neq 0$ such that $z_v = 0$ for some $v\in V$, and $L_vz_v = 0$ for all $v\in V$. By the union bound, the probability of error $p_{\text{error}}$ satisfies
\begin{align}\label{eq.union_bound}
&p_{\text{error}} \le\nonumber\\
& \sum_{\emptyset\subsetneq U\subsetneq V} \bP\left(\exists z=(z_v)_{v\in V} \text{ supported on }U \text{ s.t. } L_vz_v = 0, \forall v\right),
\end{align}
where we call that $z$ is supported on $U\subseteq V$ iff $z_u \neq 0$ for all $u\in U$ while $z_u = 0$ for all $u\notin U$. For each individual term in \eqref{eq.union_bound}, note that if $z$ is supported on $U$, then all random outcomes $R_e$ must be zero except for $(R_e)_{e\in E: e\subseteq U }$. Furthermore, for each fixed $z$ supported on $U$, the probability of $L_vz_v = 0$ for all $v$ is exactly
$$
2^{-\sum_{v\in U} Nr_v}  = 2^{-N\sum_{v\in U} r_v}. 
$$
Hence, by a union bound again, we conclude that for all $\emptyset\subsetneq U\subsetneq V$, 
\begin{align}\label{eq.pointwise_bound}
&\bP\left(\exists z=(z_v)_{v\in V} \text{ supported on }U \ \text{s.t. } L_vz_v = 0, \forall v\right) \nonumber\\
&\le 2^{\sum_{e\in E: e\subseteq U} Ns_e}\cdot 2^{-N\sum_{v\in U} r_v} \\
&= 2^{-N(\sum_{v\in U} r_v - \sum_{e\in E: e\subseteq U} s_e )}. 
\end{align}
Since all inequality constraints of the linear program are strict for $(r_v)_{v\in V}$ and $(s_e)_{e\in E}$, the above quantity is exponentially small, and \eqref{eq.union_bound}--\eqref{eq.pointwise_bound} gives $p_{\text{error}}<1$ by choosing $N$ large enough. Therefore, there exists one realization of the random matrices such that the resulting scheme is error free, as desired. 

\section{Proof of Theorem \ref{thm.decodable}}\label{subsec.correctness}
In this subsection, we show that every player may decode the random vector $X$ under the communication strategy in Definition \ref{def.strategy}, and thereby complete the proof of Theorem \ref{thm.decodable}. 

First we introduce some notations. Given the minimal topologically $k$-connected graph $G=([n],E)$, let $A$ be the incidence matrix of $G$ (as per the proof of Lemma \ref{thm.connectivity}). For linear subspaces $S,T$ of $V$, denote by $S^\perp$ the orthogonal complement of $S$, and by $S\oplus T$ the direct sum of $S$ and $T$. For any column vector $v$ and hyperedge $e\in E$, denote by $v(e)\in \FF_2$ the entry of $v$ corresponding to the hyperedge $e$. For any $(k-1)$-tuple $t\in \binom{[n]}{k-1}$, denote by $a_t$ the corresponding column vector of $A$. Note that $a_t(e)=\1(t\subseteq e)\in \FF_2$ for $e\in E$, and we will abuse notation slightly to write $a_t(e)=\1(t\subseteq e)$ for any $e\in \binom{[n]}{k}$. Finally, for any $e\in \binom{[n]}{k}$, denote by $\chi_e\in \FF_2^{|E|}$ the characteristic column vector of the hyperedge $e$ defined as $\chi_e(e')=\1(e=e')$ for any $e'\in E$.

To show that every player knows the random vector $X$, by symmetry it suffices to prove that player $1$ may decode $X$. Note that the available information for player $1$ comes from two sources: firstly, he directly knows $(R_e: 1\in e\in E)$ based on the random coins shared with him; secondly, he may see the messages $M_2,\cdots,M_n$ written by others on the blackboard. Since each bit of message corresponds to one linear equation of $X$, player $1$ may solve $X$ via a linear system of the form $BX=y$, where each entry of $y$ is either the randomness already known at player $1$ or the message written on the blackboard, and the matrix $B$ takes the form in Figure \ref{fig.B}.

\begin{figure}[h]
	\centering
	\begin{tikzpicture}[
	style1/.style={
		matrix of math nodes,
		every node/.append style={text width=#1,align=center,minimum height=5ex},
		nodes in empty cells,
		left delimiter=[,
		right delimiter=],
	}]
	\matrix[style1=0.65cm] (B)
	{
		& & & & \\
		& & & & \\
		& & & & \\
		& & & & \\
		& & & & \\
		& & & & \\
	};
	\draw [dashed] (B-1-1.south west) -- (B-1-5.south east); 
	\draw [dashed] (B-1-2.north east) -- (B-1-2.south east); 
	\draw [dashed] (B-2-1.south west) -- (B-2-5.south east); 
	\draw [dashed] (B-3-1.south west) -- (B-3-5.south east); 
	\draw [dashed] (B-4-1.south west) -- (B-4-5.south east); 
	\draw [dashed] (B-5-1.south west) -- (B-5-5.south east); 
	\node at ([xshift=-2pt] B-1-2.west) {$I$}; 
	\node at (B-1-4) {$0$}; \node at (B-2-3) {$B_2$}; 
	\node at (B-4-3) {$B_i$}; \node at (B-6-3) {$B_n$}; 
	\node at ([yshift=2pt] B-3-3) {$\vdots$}; \node at ([yshift=2pt] B-5-3) {$\vdots$}; 
	
	\draw[decoration={brace,mirror,raise=12pt},decorate] (B-1-1.north west) -- node[left=15pt] {Source I} (B-1-1.south west);
	\draw[decoration={brace,mirror,raise=12pt},decorate] (B-2-1.north west) -- node[left=15pt] {Source II} (B-6-1.south west);
	\draw[decoration={brace,raise=12pt},decorate] (B-4-5.north east) -- node[right=15pt] {Player $i$} (B-4-5.south east);
	\end{tikzpicture}
	\caption{Structure of the matrix $B$.}\label{fig.B}
\end{figure}

Clearly the number of unknowns in this linear system is $|E|=\binom{n-1}{k-1}$, and the number of linear equations is also
\begin{align*}
&|\{e\in E:1\in e\}| + \sum_{i=2}^n \left(|\{e\in E: i\in e\} - \binom{n-2}{k-2}\right) \\
&= k|E| - (n-1)\binom{n-2}{k-2} = \binom{n-1}{k-1},
\end{align*}
we conclude that $B$ is a square matrix. Hence, to prove that $BX=y$ has a unique solution $X$, it suffices to show that the matrix $B$ is of full rank, or equivalently, the row vectors of $B$ span the entire vector space $\FF_2^{|E|}$. Let $T_i\subseteq \FF_2^{|E|}$ be the row space of $B_i$ for $i\in [n]$ (where $B_1\triangleq [I,0]$), it further suffices to show that $\oplus_{i=1}^n T_i=\FF_2^{|E|}$. 

Next we characterize the vector spaces $T_i$. For $i=1$, clearly 
\begin{align}\label{eq.T1}
T_1 = \text{span}_{\FF_2}(\chi_e: 1\in e\in E) = [\text{span}_{\FF_2}(\chi_e: 1\notin e\in E)]^\perp. 
\end{align}
For $i>1$, let $A_i$ be the incidence matrix of the induced hypergraph $G_i$ (an illustration is shown in Figure \ref{fig.A}, with $A'$ replaced by $A_i$). By the construction of the strategy in Definition \ref{def.strategy}, each row of $B_i$ corresponds to some selection of rows in $A_i$ such that the selected rows sum into zero. Moreover, since player $i$ does not know $(R_e: i\notin e\in E)$ when writing on the blackboard, each row of $B_i$ is also supported on $(e\in E: i\in e)$. Hence, the restriction of rows of $B_i$ on the coordinates $\{e\in E: i\in e\}$ exactly span the nullspace of $A_i$, regardless of the choice of the minimal $(k-1)$-connected subgraph $G_i^\star$. Adding the support constraint together, we conclude that
\begin{align}\label{eq.Ti}
T_i = \left[\text{span}_{\FF_2}\left(\left(a_t: i\in t\in \binom{[n]}{k-1}\right), \left(\chi_e: i\notin e\in E\right)\right)\right]^\perp
\end{align}
for $i>1$. By \eqref{eq.T1} and \eqref{eq.Ti}, writing
\begin{align*}
S_1 &= \text{span}_{\FF_2}(\chi_e: 1\notin e\in E), \\
S_i &= \text{span}_{\FF_2}\left(\left(a_t: i\in t\in \binom{[n]}{k-1}\right), \left(\chi_e: i\notin e\in E\right)\right)
\end{align*}
for $i>1$, the identity $(\oplus_{i=1}^n T_i)^\perp = \cap_{i=1}^n T_i^\perp$ implies that the desired result $\oplus_{i=1}^n T_i=\FF_2^{|E|}$ is further equivalent to $\cap_{i=1}^n S_i = \{0\}$. 

Now suppose that $v\in \cap_{i=1}^n S_i$, then by definitions of $S_i$, we may write 
\begin{align}\label{eq.expansion}
v = \sum_{e\in E: 1\notin e} \beta_e^{(1)}\chi_e = \sum_{t: i\in t} \alpha_t^{(i)} a_t + \sum_{e\in E: i\notin e} \beta_e^{(i)}\chi_e, \qquad \forall i>1,
\end{align}
where $\alpha_t^{(i)}, \beta_e^{(i)}\in \FF_2$ are some binary coefficients. We may define $\alpha_t^{(1)}=0$ for any $t\ni 1$ to make \eqref{eq.expansion} symmetric in $i\in [n]$. Now for any hyperedge $e^\star=(i_1,\cdots,i_k)\in E$, evaluating both sides of \eqref{eq.expansion} at coordinate $e^\star$ yields
\begin{align}\label{eq.equation_alpha}
\sum_{t: i_j\in t\subseteq e^\star} \alpha_t^{(i_j)}  = v(e^\star), \qquad \forall j\in [k]. 
\end{align}
As a result, we have arrived at another system of linear equations with unknowns $(\alpha_t^{(i)}: i\in t)$ and $(v(e): e\in E)$. The number of unknowns for this system is
\begin{align*}
\binom{n}{k-1}\cdot (k-1) + |E| = \frac{(n-1)k+1}{n-k+1}\cdot\binom{n-1}{k-1}. 
\end{align*}
However, the number of linear equations of type \eqref{eq.equation_alpha} is only $k|E|$, and we need an additional number of
\begin{align*}
\frac{(n-1)k+1}{n-k+1}\cdot\binom{n-1}{k-1} - k|E| = (k-1)\cdot\binom{n-1}{k-2}
\end{align*}
boundary conditions. We claim that the boundary condition can be $\alpha_t^{(i)}=0$ whenever $1\in t$. For $i=1$, this is simply our special treatment for the player $1$. For $i>1$, we need the following lemma. 
\begin{lemma}\label{lemma.column_basis}
	Let $G$ be a minimal topologically $k$-connected hypergraph with incidence matrix $A$. Then the column vectors $(a_t: 1\notin t)$ constitute a linearly independent column basis of $A$. 
\end{lemma}
\begin{proof}
	Since $\text{rank}(G) = \binom{n-1}{k-1}=|\{t\in \binom{[n]}{k-1}: 1\notin t\}|$, it suffices to prove that the column vectors $(a_t: 1\notin t)$ are linearly independent over $\FF_2$. Suppose that $\sum_{t: 1\notin t} \alpha_t a_t = 0$ for coefficients $\alpha_t\in \FF_2$, evaluating both sides at hyperedge $e\in E$ yields
	\begin{align*}
	\sum_{t: 1\notin t} \alpha_t a_t(e) = 0,\qquad \forall e\in E. 
	\end{align*}
	Recall that we have slightly abused the notation and defined $a_t(e)=\1(t\subseteq e)$ for any $e\in \binom{[n]}{k}$. Under the general notation, if the hyperedge $e$ is generated by $e_1,\cdots,e_m\in E$, then
	\begin{align}\label{eq.generation}
	\sum_{i=1}^m a_t(e_i) = a_t(e). 
	\end{align}
	In fact, \eqref{eq.generation} can be shown by comparing the number of occurrences of each $(k-1)$-tuple $t$ at both sides, and the generation step in Definition \ref{def.connectivity} ensures that they are of the same parity. With the help of \eqref{eq.generation}, and using the fact that $G$ is topologically $k$-connected, we have
	\begin{align*}
	\sum_{t: 1\notin t} \alpha_t a_t(e) = 0,\qquad \forall e\in \binom{[n]}{k}. 
	\end{align*}
	Now for any $t^\star\in \binom{[n]}{k-1}$, choosing $e^\star = t^\star \cup \{1\}$ in the previous identity yields to $\alpha_{t^\star}=0$, which proves the desired linear independence. 
\end{proof}
\begin{remark}
	Lemma \ref{lemma.column_basis} is the first occurrence where we \emph{require} that $G$ is topologically $k$-connected, while previously we only \emph{assume} this property without really using it. The key to this property is equation \eqref{eq.generation}, which implies that as long as some linear equations of column vectors $a_t$ hold for all $e\in E$, it will hold for any $k$ tuples $e\in \binom{[n]}{k}$. 	
\end{remark}

Applying Lemma \ref{lemma.column_basis} to the incidence matrix of the induced hypergraphs (i.e., the matrix $A'$ in Figure \ref{fig.A}), we conclude that the column vectors $(a_t: i\in t,1\notin t)$ is a linearly independent basis of $(a_t: i\in t)$. Therefore, we may set $\alpha_t^{(i)}=0$ whenever $1\in t$ in \eqref{eq.equation_alpha} to remove the redundant variables.

Let the vector $\gamma$ be the collection of all unknowns $\alpha_t^{(i)}$ and $v(e)$, by the previous discussion, we arrive at a system of linear equations $D\gamma = 0$, where $D$ is a square matrix. Specifically, the top rows of $D$ constitute the identity matrix concatenated with zeros corresponding to the boundary conditions $\alpha_t^{(i)}=0$ whenever $1\in t$. For other rows, each $e=(i_1,\cdots,i_k)\in E$ (where possibly $1\in e$) gives rise to $k$ linear equations of the form \eqref{eq.equation_alpha}, where $v(e)$ appears in all equations, and the variables $\alpha_t^{(i_j)}$ only appear in one equation for each $j\in [k]$. A pictorial illustration of the previous structures is shown in Figure \ref{fig.D}. 

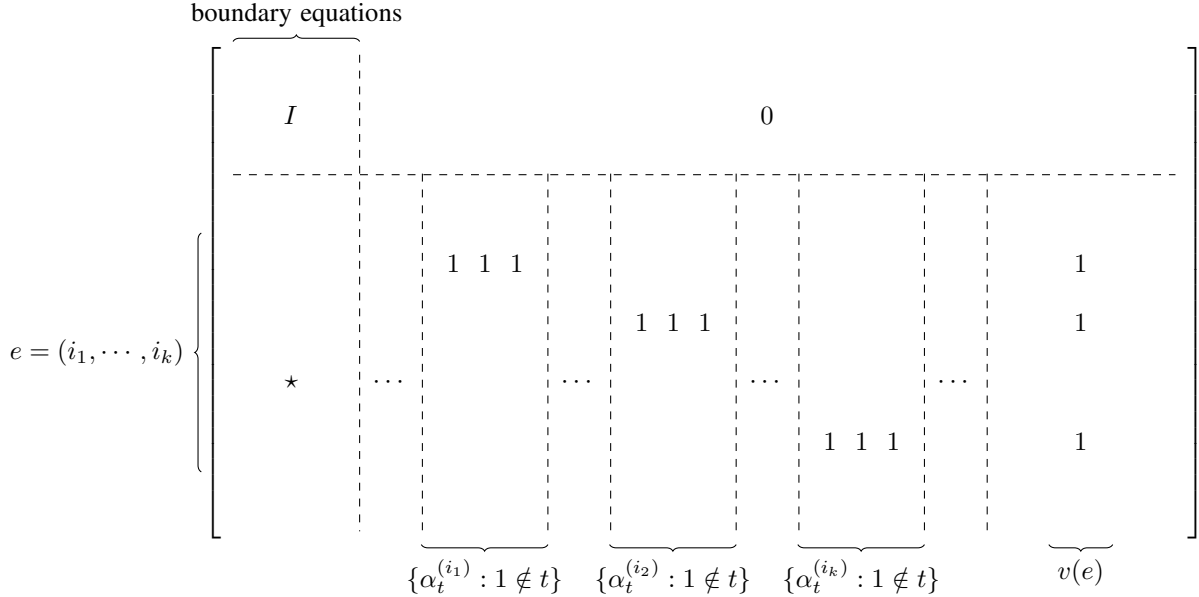
\begin{figure*}[htbp]
	\centering
	\begin{tikzpicture}[
	style1/.style={
		matrix of math nodes,
		every node/.append style={text width=#1,align=center,minimum height=5ex},
		nodes in empty cells,
		left delimiter=[,
		right delimiter=],
	}]
	\matrix[style1=0.6cm] (D)
	{
		& & & & & & & & & & & & & & \\ 
		& & & & & & & & & & & & & & \\ 
		& & & & & & & & & & & & & & \\ 
		& & & & & & & & & & & & & & \\ 
		& & & & & & & & & & & & & & \\ 
		& & & & & & & & & & & & & & \\ 
		& & & & & & & & & & & & & & \\ 
		& & & & & & & & & & & & & & \\ 
	};
	\draw [dashed] (D-2-1.south west) -- (D-2-15.south east); 
	\draw [dashed] (D-1-2.north east) -- (D-8-2.south east); 
	\draw [dashed] (D-3-3.north east) -- (D-8-3.south east); 
	\draw [dashed] (D-3-5.north east) -- (D-8-5.south east); 
	\draw [dashed] (D-3-6.north east) -- (D-8-6.south east); 
	\draw [dashed] (D-3-8.north east) -- (D-8-8.south east); 
	\draw [dashed] (D-3-9.north east) -- (D-8-9.south east); 
	\draw [dashed] (D-3-11.north east) -- (D-8-11.south east); 
	\draw [dashed] (D-3-12.north east) -- (D-8-12.south east);
	\node at ([xshift=-2pt] D-1-2.south west) {$I$}; 
	\node at (D-1-9.south) {$0$}; 
	\node at ([xshift=-2pt] D-6-1.east) {$\star$};
	\node at (D-6-3) {$\cdots$}; 
	\node at (D-6-6) {$\cdots$}; 
	\node at (D-6-9) {$\cdots$}; 
	\node at (D-6-12) {$\cdots$}; 
	\node at (D-4-4.east) {$1\ \ 1\ \ 1$}; 
	\node at (D-5-7.east) {$1\ \ 1\ \ 1$}; 
	\node at (D-7-10.east) {$1\ \ 1\ \ 1$}; 
	\node at (D-4-14) {$1$}; \node at (D-5-14) {$1$}; \node at (D-7-14) {$1$}; 
	
	\draw[decoration={brace,raise=5pt},decorate] (D-1-1.north west) -- node[above=7pt] {boundary equations} (D-1-2.north east);	
	\draw[decoration={brace,mirror,raise=12pt},decorate] (D-4-1.north west) -- node[left=15pt]  {$e=(i_1,\cdots,i_k)$} (D-7-1.south west);
	\draw[decoration={brace,mirror,raise=5pt},decorate] (D-8-4.south west) -- node[below=7pt]  {$\{\alpha_t^{(i_1)}: 1\notin t\}$} (D-8-5.south east);
	\draw[decoration={brace,mirror,raise=5pt},decorate] (D-8-7.south west) -- node[below=7pt]  {$\{\alpha_t^{(i_2)}: 1\notin t\}$} (D-8-8.south east);
	\draw[decoration={brace,mirror,raise=5pt},decorate] (D-8-10.south west) -- node[below=7pt]  {$\{\alpha_t^{(i_k)}: 1\notin t\}$} (D-8-11.south east);
	\draw[decoration={brace,mirror,raise=5pt},decorate] (D-8-14.south west) -- node[below=7pt]  {$v(e)$} (D-8-14.south east);
	\end{tikzpicture}
	\caption{Structure of the matrix $D$.}\label{fig.D}
\end{figure*}

Note that it remains to prove that $\gamma=0$, it suffices to show that $D$ is of full rank. Let $D^\star$ be the sub-matrix of $D$ at the lower right corner of Figure \ref{fig.D}, it further suffices to prove that $D^\star$ is of full rank, and in particular, the columns of $D^\star$ are linearly independent over $\FF_2$. Let $(d_t^{(i)}: 1\notin t, i\in t)$ and $(d_{v(e)}: e\in E)$ be the column vectors of $D^\star$, and for each $e\in E$, we overload our notation $v(e)$ to denote the $k$-dimensional projection of the column vector $v$ to the $k$ coordinates corresponding to $e$. Suppose that
\begin{align}\label{eq.linear_ind_D}
0 = \sum_{i=2}^n \sum_{t: i\in t, 1\notin t} \delta_t^{(i)}d_t^{(i)} + \sum_{e\in E} \delta_e d_{v(e)}
\end{align}
holds for some coefficients $\delta_t^{(i)}, \delta_e\in \FF_2$. Note that for $e\in E$, we have
\begin{align}
&d_t^{(i)}(e) \in \left\{ \left[\begin{matrix}
0 \\
0 \\
\vdots \\
0
\end{matrix}\right], \left[\begin{matrix}
1 \\
0 \\
\vdots \\
0 
\end{matrix}\right], \left[\begin{matrix}
0 \\
1 \\
\vdots \\
0 
\end{matrix}\right], \cdots, \left[\begin{matrix}
0 \\
0 \\
\vdots \\
1 
\end{matrix}\right]
\right\},
\nonumber\\
 & d_{v(e')}(e)\in \left\{ \left[\begin{matrix}
0 \\
0 \\
\vdots \\
0
\end{matrix}\right], \left[\begin{matrix}
1 \\
1 \\
\vdots \\
1 
\end{matrix}\right] \right\}\label{eq.projection}. 
\end{align}
In fact, we may write $d_t^{(i)}(e)=a_t(e)\cdot e_{j_i(t)}$, where $a_t(e)=\1(t\subseteq e)$ is the evaluation of the $t$-th column vector of the incidence matrix $A$ on the vertex $v$, and $e_j$ is the $j$-th canonical vector of $\FF_2^k$. Note that the index $j_i(t)$ only depends on the choice of the permutation of elements of $e$, and thus $j_i(t)\neq j_{i'}(t)$ for $i\neq i'\in t$. By equality \eqref{eq.generation} and the topological $k$-connectivity of $G$, we may evaluate both sides of \eqref{eq.linear_ind_D} on all $e\in \binom{[n]}{k}$, with projections of column vectors given by \eqref{eq.projection}. Hence, given any $t^\star\in \binom{[n]}{k-1}$ with $1\notin t$, we may form the hyperedge $e^\star=t^\star\cup \{1\}$, and evaluating $e^\star$ on both sides of \eqref{eq.linear_ind_D} yields
\begin{align}\label{eq.final_equality}
0 = \sum_{i\in t^\star} \delta_{t^\star}^{(i)} d_{t^\star}^{(i)}(e^\star) + c\left[\begin{matrix}
1 \\
1 \\
\vdots \\
1 
\end{matrix}\right],
\end{align}
where $c\in \FF_2$ is some scalar. By our previous discussion, there are $(k-1)$ terms in the summation, each of which is some canonical vector in $\FF_2^k$ with coefficient $\delta_{t^\star}^{(i)}$. Moreover, these canonical vectors (for different $i\in t^\star$) must be different. Hence, in order for \eqref{eq.final_equality} to hold, we must have $\delta_{t^\star}^{(i)}=0$ for all $i\in t^\star$ and $c=0$. By the arbitrariness of our choice of $t^\star$, we conclude that all coefficients in \eqref{eq.linear_ind_D} are zero, and thus $D^\star$ is linearly independent. Therefore, we have shown that every player may decode the random vector $X$ under the strategy in Definition \ref{def.strategy}, and thus completed the proof of Theorem \ref{thm.decodable}. 

\section{Proof of Theorem \ref{thm.cluster}}\label{theorem4}
	Firstly we compute $H(X)$ and $H(M)$ to verify that this strategy achieves the optimal communication rate. In $i$-th connected component, the strategy in Definition \ref{def.strategy} is employed $M_i$ times, and thus
\begin{align}\label{eq.HX}
H(X) &= \sum_{i=1}^m M_i\cdot \binom{|A_i|-1}{k-1} = \frac{C}{k-1}\sum_{i=1}^m (|A_i| - 1) \\
&= \frac{C(n-1)}{k-1},
\end{align}
where we have used Lemma \ref{lemma.path_connectivity_edge} in the last step. Similarly, summing the messages within components and across components, we arrive at
\begin{align}
H(M) &= \sum_{i=1}^m M_i\cdot \binom{|A_i|-2}{k-1} + \sum_{v\in V} C\cdot (\deg_{G_c}(v)-1) \nonumber \\
&= \frac{C}{k-1} \sum_{i=1}^m (|A_i| - k) + C\sum_{v\in V} (\deg_{G_c}(v)-1)\nonumber \\
&=  \frac{C(n-k)}{k-1}\label{eq.HM},
\end{align}
where \eqref{eq.HM} follows from both statements of Lemma \ref{lemma.path_connectivity_edge}. Combining \eqref{eq.HX} and \eqref{eq.HM}, we arrive at the desired communication rate. 

It remains to show that every player may decode the entire vector $X$ based on his own information and messages written on the blackboard. First we recall the following fact: for a topologically $k$-connected hypergraph $G=(V,E)$, a new player who is not in this hypergraph can decode all outcomes after seeing the messages on the blackboard following the strategy in Definition \ref{def.strategy}, as well as all coin tossing outcomes corresponding to edges of $G_v^\star$ (cf. Definition \ref{def.strategy}) for an \emph{arbitrary} player $v\in V$. In fact, using the additional information in $G_v^\star$ together with the messages $v$ writes on the blackboard, by the rules in Definition \ref{def.strategy}, the new player can decode the outcomes of all coins shared with $v$. Hence, the new player is effectively ``equivalent to'' $v$ in the sense that they have the same observations, and the new player can decode all outcomes (as $v$ can) by the proof in Section \ref{subsec.correctness}. 

By symmetry it suffices to show that any player $v_1\in A_1$ may decode the entire vector $X$. Firstly, by Theorem \ref{thm.main.upper} and the messages within the component $A_1$, the player $v_1$ can decode all outcomes in the component $A_1$. Since the hypergraph $G_c$ is path-connected, the component $A_1$ must intersect with other components, say $A_2$, at some point $v_2$. Now by the messages across the components $A_1$ and $A_2$ written by $v_2$, the player $v_1$ knows all coin tossing outcomes corresponding to edges of $G_{v_2}^\star$ in the component $A_2$. By the previous fact, now $v_1$ can decode all outcomes in the component $A_2$. This process may continue to cover all connected components due to the path connectivity of $G_c$, and we conclude that $v_1$ can decode the entire outcome vector $X$, as claimed. 

\section{Proofs of Main Lemmas}\label{appendix.proof_lemma}
\subsection{Proof of Lemma \ref{thm.connectivity}}\label{lemma1}
	For a $k$-uniform hypergraph $G=(V,E)$, define the following version of the incidence matrix $A$ of $G$: each row of $A$ corresponds to a hyperedge $e\in E$, and each column of $A$ corresponds to a $(k-1)$-tuple in $[n]$. The entries of $A$ are defined as 
	\begin{align*}
	A_{e,t} = \1(t\subseteq e) \in \FF_2, \qquad e\in E, t\in \binom{[n]}{k-1}. 
	\end{align*}
	Hence, the dimension of $A$ is $|E|\times \binom{n}{k-1}$ (see Figure \ref{fig.incidence_matrix} for an example). 
	
	\begin{figure}[!ht]
	\scalebox{0.77}{
		$$
		\begin{blockarray}{ccccccccccc}
		& (12) & (13) & (14) & (15) & (23) & (24) & (25) & (34) & (35) & (45) \\
		\begin{block}{c[cccccccccc]}
		(123) & 1 & 1 &  &  & 1 &  &  &  &  & \bigstrut[t] \\
		(124) & 1 &  &  1&  &  & 1 &  &  &   \\
		(134)&  & 1 & 1 &  &  &  &  & 1 &  \\
		(125)& 1 &  &  & 1 & &  &  1&  &  \\
		(235)& &  &  &  & 1 &  & 1 &  & 1 \\
		(245) &  &  &  &  & & 1 & 1 &  & &1\bigstrut[b]\\
		\end{block}
		\end{blockarray}\vspace*{-1.25\baselineskip}
		$$
		}
		\caption{Incidence matrix of the hypergraph in Figure \ref{fig.connected_3hypergraph}.}\label{fig.incidence_matrix}
	\end{figure}
	
	According to the definition of topological $k$-connectivity, a hyperedge $e$ can be generated by hyperedges $e_1,\cdots,e_m$ if and only if the rows corresponding to $e,e_1,\cdots,e_m$ sum into the zero vector in $\FF_2$. Let $A^\star$ be the incidence matrix of the complete $k$-uniform hypergraph, then a minimal topologically $k$-connected hypergraph is simply a linearly independent basis of the row vectors of $A^\star$. Hence, the number of hyperedges in any minimal topologically $k$-connected hypergraph is $\text{rank}(A^\star)$. 
	
	Consider the incidence matrix $A$ of a star graph, i.e., $E=\{e\in \binom{[n]}{k}: 1\in e\}$. We show that the rows of $A$ are linearly independent: for any tuple $t\in \binom{[n]}{k-1}$ with $1\notin t$, there is only one hyperedge of $A$ which contains $t$. Furthermore, any hyperedge $e\in \binom{[n]}{k}$ in the complete $k$-uniform hypergraph can be generated from this star graph: clearly $e\in E$ if $1\in e$, and $e$ can be generated by $e_1,\cdots,e_k$ if $1\notin e$, where $e_i=e\cup \{1\}\backslash \{i\text{-th element of }e\}$. Hence the rows of $A$ constitute a linearly independent basis of $A^\star$, and 
	\begin{align*}
	\text{rank}(A^\star) = |E| = \binom{n-1}{k-1}, 
	\end{align*}
	as desired. 

\subsection{Proof of Lemma \ref{lemma.induced_hypergraph}}\label{lemma2}
	It suffices to prove that $G_1$ is topologically $(k-1)$-connected, and the proof relies on linear algebra. Let $A$ be the incidence matrix of $G$ (as per the proof of Lemma \ref{thm.connectivity}), and $A'$ be the sub-matrix of $A$ consisting of rows (hyperedges) $e\ni 1$ and columns (tuples) $t\ni 1$. Relabeling the rows and columns of $A'$ by removing the common element $1$ in the indices, it is clear that $A'$ is the incidence matrix of $G_1$. A pictorial illustration is displayed in Figure \ref{fig.A}. 
	\begin{figure}[h]
		\centering
		\begin{tikzpicture}[
		style1/.style={
			matrix of math nodes,
			every node/.append style={text width=#1,align=center,minimum height=5ex},
			nodes in empty cells,
			left delimiter=[,
			right delimiter=],
		}]
		\matrix[style1=0.65cm] (A)
		{
			& & & & \\
			& & & & \\
			& & & & \\
			& & & & \\
		};
		\draw [dashed] (A-1-2.north east) -- (A-4-2.south east); 
		\draw [dashed] (A-2-1.south west) -- (A-2-2.south east); 
		\node at ([xshift=-15pt, yshift=-10pt] A-1-2) {$A'$};
		\node at ([xshift=-15pt, yshift=-12pt] A-3-2) {$0$};
		\node at ([yshift = -10pt] A-2-4) {$\star$}; 
		
		\draw[decoration={brace,mirror,raise=12pt},decorate] (A-1-1.north west) -- node[left=15pt] {$\{e\in E: 1\in e\}$} (A-2-1.south west);
		\draw[decoration={brace,raise=5pt},decorate] (A-1-1.north west) -- node[above=8pt] {$\{t\in \binom{[n]}{k-1}: 1\in t\}$} (A-1-2.north east);
		\draw[decoration={brace,raise=12pt},decorate] (A-1-5.north east) -- node[right=15pt] {$A$} (A-4-5.south east);
		\end{tikzpicture}
		\caption{An illustration of matrices $A$ and $A'$.}\label{fig.A}
	\end{figure}
	
	To show that $G_1$ is topologically $(k-1)$-connected, it is equivalent to show that the row space of $A'$ contains all $r_{e'}$ for $e'\in \binom{[n]\backslash \{1\}}{k-1}$, where $r_{e'}$ is the row vector corresponding to the hyperedge $e'$. Note that each $r_{e'}$ gives rise to a row vector $r_e$ for the original hypergraph $G$, with $e=e'\cup \{1\}$. Since $G$ is $k$-connected, the row vector $r_e$ can be written as the sum of some rows of $A$. Restricting to rows $\{e\in E: 1\in e\}$, it is clear from the pictorial illustration that the corresponding rows of $A'$ will sum into $r_{e'}$, as desired. 

\subsection{Proof of Lemma \ref{lemma.path_connectivity_edge}}\label{lemma3}
	We prove the first statement by induction on $|E|$. For the base case, if $E=\{A\}$ only consists of one hyperedge, then the path connectivity ensures $A=V$, and the result is obvious. Now suppose that the results holds for any hypergraph $G=(V,E)$ with $|E|<m$. We first show that there cannot be two hyperedges $A_1, A_2\in E$ such that $|A_1 \cap A_2| > 1$ in the cycle-free hypergraph $G$. In fact, if $u,v\in A_1\cap A_2$, then $u\overset{A_1}{\to} v\overset{A_2}{\to} u$ is a simple cycle in $G$, a contradiction. Hence, any two hyperedges $A_1, A_2$ are either disjoint or intersecting at one vertex. 
	
	Next we show that there must be a \emph{leaf} hyperedge in $G$, where $A\in E$ is defined to be a leaf hyperedge iff $|A \cap \left(\cup_{B\in E \backslash \{A\} } B\right)|=1$. Start from any hyperedge $A_0 \in E$: if $A_0$ is a leaf hyperedge, we are done. Otherwise, by path connectivity there must be some $v_0\in A_0$ and $A_1 \in E \backslash \{A_0\}$ such that $v_0\in A_1$. We are done if $A_1$ is a leaf hyperedge, and otherwise $A_1$ intersects with other hyperedges at more than one point, i.e., we may find some $v_1\in A_1 \backslash \{v_0\}, A_2\in E \backslash \{A_0,A_1\}$ such that $v_1\in A_2$. Continuing this process, we either arrive at some leaf hyperedge, or find some $v_k=v_{\ell}$ with $k<\ell$ in this process. The latter case is impossible, for $v_k\overset{A_{k+1}}{\to} v_{k+1} \overset{A_{k+2}}{\to} \cdots \overset{A_{\ell}}{\to} v_\ell$ is a cycle in $G$. Therefore, there must be a leaf hyperedge $A$ in $G$. 
	
	Now remove $A$ and all isolated $|A|-1$ vertices from $G$. It is straightforward to see that the remaining hypergraph is still path-connected and cycle-free, then by induction hypothesis 
	\begin{align*}
	\sum_{B\in E - \{A\}} (|B| - 1) = |V| - (|A|-1) - 1. 
	\end{align*} 
	Rearranging gives the desired result. 

For the second statement, by a double counting argument we have
		\begin{align*}
		\sum_{v\in V} \deg(v) &= \sum_{v\in V} \sum_{A\in E} \1(v\in A)\\
		& =  \sum_{A\in E} \sum_{v\in V} \1(v\in A)\\
		&= \sum_{A\in E} |A|. 
		\end{align*}
		Now the desired inequality follows from \cref{lemma.path_connectivity_edge}.

\balance
\section{Proof of Theorem \ref{thm.star-graph}}\label{appendix.star_graphs}
\subsection{The \emph{if} part}
We first prove the \emph{if} part by providing an explicit communication strategy. Without loss of generality we assume that the induced graph $G_{v^\star}$ is exactly a simple edge or a Hamilton cycle of odd length, as the general disjoint union can be handled in exactly the same way as Definition \ref{def.strategy_cluster}. Further, if $G_{v^\star}$ is a simple edge, then $G$ is a triangle and there is nothing to prove. Hence, it remains to consider the case where $G_{v^\star}$ is a Hamilton cycle of odd length: 
\begin{align*}
&G_{v^\star} =\\
& \{[2m+1], \{(1,2), (2,3), \cdots, (2m,2m+1), (2m+1,1) \}  \}. 
\end{align*}
By definition of induced graphs in Section \ref{subsec.strategy}, we may use $R_{1,2}$ to denote the randomness associated with $R_{(1,2,v^\star)}$ in the original star graph, and similarly for others. 

The communication strategy is as follows. The central node $v^\star$ writes the following three sets of messages on the blackboard: 
\begin{align*}
M_1 &= \{ R_{1,2} \oplus R_{3,4}, R_{1,2} \oplus R_{5,6}, \ldots, R_{1,2} \oplus R_{2m-1,2m} \}, \\
M_2 &= \{ R_{2,3} \oplus R_{4,5}, R_{2,3} \oplus R_{6,7}, \ldots, R_{2,3} \oplus R_{2m,2m+1} \}, \\
M_3 &= \{ R_{1,2} \oplus R_{2,3} \oplus R_{2m+1,1} \}. 
\end{align*}
Note that there are $2m-1=n-3$ bits in the message $M_1\cup M_2\cup M_3$, and the total amount of randomness is $2m+1 = n-1$ bits. Hence, if each player can decode all listed random bits then the communication rate is optimal. This can be easily shown as follows: first, the central vertex $v^\star$ knows all random bits; second, the special player $2$ can decode all other random bits directly as all messages involves either $R_{1,2}$ or $R_{2,3}$; finally, all other players can decode $R_{2,3}$ based on $M_2$ and $R_{1,2}$ based on $M_1$ (the player $2m+1$ additionally requires $M_3$), and are therefore as informative as the player $2$. The above arguments show that all players are able to decode all random bits, and therefore complete the proof of the \emph{if} part of Theorem \ref{thm.star-graph}. 

\subsection{The \emph{only if} part}
The \emph{only if} part is slightly more challenging. First, by the proof of Corollary \ref{cor.lower}, the assumption that the optimal communication rate $(n-3)/(n-1)$ is achievable implies the existence of non-negative scores $r_v$ assigned to each vertex $v$ and $s_e$ assigned to each hyperedges $e$ such that
\begin{align}
\sum_{v\in V} r_v &= \frac{n-3}{n-1}, \label{eq:sum_r} \\
\sum_{e\in E} s_e &= 1, \label{eq:sum_s} \\
\sum_{v\in V - \{v_0\} } r_v &= \sum_{e\in E: e\subseteq V-\{v_0\}  } s_e, \qquad \forall v_0\in V. \label{eq:equal_rs}
\end{align}
Choosing $v_0 = v^\star$ in \eqref{eq:equal_rs}, the RHS is zero, and the non-negativity of $r_v$ implies that $r_v = 0$ for all $v\neq v^\star$. Further, \eqref{eq:sum_r} shows that $r_{v^\star} = (n-3)/(n-1)$. Now choosing any $v_0\in V- \{v^\star\}$ in \eqref{eq:equal_rs} leads to
\begin{align*}
\sum_{e\in E: e\subseteq V-\{v_0\}  } s_e = \frac{n-3}{n-1}, 
\end{align*}
which together with \eqref{eq:sum_s} gives
\begin{align}\label{eq:fractional_match}
\sum_{e\in E: v_0\in e  } s_e = \frac{2}{n-1}, \qquad \forall v_0\in V - \{v^\star\}. 
\end{align}

Now we relate the condition \eqref{eq:fractional_match} to the notion of fractional matchings in fractional graph theory. Let $G = (V,E)$ be a classical graph (not a hypergraph), a \emph{fractional matching} $f$ of $G$ is an assignment $\{f(e)\}_{e\in E}$ to all edges of $G$ such that $f(e)\ge 0$ for all $e\in E$, $\sum_{e: v\in e} f(e) \le 1$ holds for all $v\in V$, and $\sum_{e\in E}f(e) = |V|/2$. To see the relationship, consider the induced graph $G_{v^\star}$ which is a classical graph, and since $G$ is a star graph, there is a bijection between $E(G)$ and $E(G_{v^\star})$. Hence, if we do not distinguish between $e\in E(G)$ and $e\in E(G_{v^\star})$, we may define $f(e) = (n-1)s_e/2$ for all $e\in E(G_{v^\star})$. We claim that $f$ is a fractional matching of the graph $G_{v^\star}$: in fact, \eqref{eq:fractional_match} shows that
\begin{align*}
\sum_{e: v_0\in e} f(e) = \frac{n-1}{2}\cdot \sum_{e\in E: v_0\in e  } s_e = 1, \qquad \forall v_0\in V(G_{v^\star}),
\end{align*}
and \eqref{eq:sum_s} shows that $\sum_{e\in E(G_{v^\star})} f(e) = (n-1)/2\cdot \sum_{e\in E} s_e = (n-1)/2 = |V(G_{v^\star})|/2$. Then the claimed result follows from the following fractional Tutte's theorem \cite[Proposition 2.2.2]{scheinerman2011fractional} which provides a necessary and sufficient condition for the existence of a fractional matching. 
\begin{theorem}
A simple graph $G$ has a fractional matching if and only if $G$ contains a vertex-disjoint union of simple edges or Hamilton cycles of odd length including all vertices. 
\end{theorem}

\bibliographystyle{IEEEtran}
\bibliography{di,Bibliography}

\begin{IEEEbiographynophoto}{Yanjun Han}
	(Member, IEEE) received the B.Eng. degree (Hons.) in electronic engineering from Tsinghua University, Beijing, China, in 2015, and the M.S. and Ph.D. degrees from Stanford University in 2017 and 2021, respectively. He is currently a Post-Doctoral Scholar with the Simons Institute for the Theory of Computing, University of California at Berkeley, Berkeley. His research interests include statistical machine learning, high-dimensional and nonparametric statistics, information theory, online learning and bandits, and their applications.
\end{IEEEbiographynophoto}

\begin{IEEEbiographynophoto}{Kedar Tatwawadi}
	received the BTech and MTech in electrical engineering from Indian Institute of Technology Bombay, India, in 2014, and the M.S. and Ph.D. degrees from Stanford University in 2017 and 2020, respectively. He is currently a Research Scientist at WaveOne Inc. His research interests include data compression, learned video compression, information theory, and their applications.
\end{IEEEbiographynophoto}

\begin{IEEEbiographynophoto}
{Gowtham R. Kurri} (Member, IEEE) graduated from the International Institute of Information Technology, Hyderabad, India, with a B.\ Tech.\ degree in Electronics and Communication Engineering, in 2011. He received his M.Sc. and Ph.D. degrees from the Tata Institute of Fundamental Research, Mumbai, India in 2020. He is currently a Post-Doctoral Researcher at the School of Electrical, Computer and Energy Engineering at Arizona State University.

From 2011-2012, he worked as an Associate Engineer at Qualcomm India Private Limited, Hyderabad, India. From July to October, 2019, he was a Research Intern in the Blockchain Technology Group at IBM Research, Bangalore, India. 
\end{IEEEbiographynophoto}

\begin{IEEEbiographynophoto}
{Zhengqing Zhou} received the B.S. degree in mathematics from University of Science and Technology of China in 2016, and the Ph.D. degree in mathematics from Stanford University in 2021. His research interests include distributionally robust optimization, statistical machine learning, and applied probability. 
\end{IEEEbiographynophoto}

\begin{IEEEbiographynophoto}
{Vinod M. Prabhakaran} (Member, IEEE) received the M.E. degree from the Indian Institute of Science in 2001 and the Ph.D. degree from the University of California, Berkeley in 2007. He was a Post-Doctoral Researcher at the Coordinated Science Laboratory, University of Illinois, Urbana-Champaign from 2008 to 2010 and at Ecole Polytechnique Fédérale de Lausanne, Switzerland in 2011. Since 2011, he has been at the School of Technology and Computer Science at the Tata Institute of Fundamental Research, Mumbai. His research interests are in information theory, communication, cryptography, and signal processing.

He has received the Tong Leong Lim Pre-Doctoral Prize and the Demetri Angelakos Memorial Achievement Award from the EECS Department, University of California, Berkeley, and the Ramanujan Fellowship from the Department of Science and Technology, Government of India. He was an Associate Editor for IEEE TRANSACTIONS
ON INFORMATION THEORY during 2016-19.
\end{IEEEbiographynophoto}

\begin{IEEEbiographynophoto}
{Tsachy Weissman} (Fellow, IEEE) has been on the faculty of the Electrical Engineering department at Stanford since 2003, conducting research in and teaching the science of information, with applications spanning genomics, neuroscience, and technology.  He has served and still does on editorial boards for scientific journals, technical advisory boards in industry, and as founding director of the Stanford Compression Forum. His recent initiatives at Stanford include the STEM2SHTEM science and humanities high school internship program, and Stagecast, a low-latency video platform allowing actors and singers to perform together in real-time while geographically distributed. IEEE fellow, he has received multiple awards for his research and teaching, including best paper awards from the IEEE Information Theory and the Communications societies, and best student authored paper awards in the top conferences of his areas of scholarship. He has prototyped some of Guardant Health's first algorithms for early detection of cancer from blood tests, and has more recently co-founded and sold Compressable to Amazon, where he is now working to reduce humanity's carbon footprint by compressing its data. His favorite gig to date was being an advisor to the HBO show “Silicon Valley”.
\end{IEEEbiographynophoto}
\end{document}